\newif\iflong
\longtrue

\iflong
\documentclass[11pt]{article}
\usepackage[papersize={8.5in,11in},margin=1in]{geometry}
\else
\documentclass[letterpaper,USenglish,numberwithinsect]{lipics}
\usepackage{chngcntr}
\fi
\usepackage{amsthm}
\usepackage{amsmath}
\usepackage{amssymb}
\usepackage{enumitem}
\usepackage{algorithmicx}
\usepackage{algpseudocode}
\usepackage{comment}
\usepackage{lmodern}
\usepackage[T1]{fontenc}
\usepackage{color}
\usepackage{microtype}
\usepackage{graphicx}
\usepackage{array} 
\usepackage{pifont} 
\newcommand{\dtlinkcolor}{{0.8 0.8 1}} 
\usepackage[hyperindex=true,pdfpagemode=UseOutlines,bookmarksnumbered=true,bookmarksopen=true,bookmarksopenlevel=2,pdfstartview=FitH,pdfborder={0 0 1},linkbordercolor=\dtlinkcolor,citebordercolor=\dtlinkcolor,urlbordercolor=\dtlinkcolor,pagebordercolor=\dtlinkcolor]{hyperref}
\hypersetup{pageanchor=false,pdfpagelabels}
\theoremstyle{slplain}
\iflong
\newtheorem{theorem}{Theorem}[section]
\newtheorem{lemma}[theorem]{Lemma}

\fi

\newtheorem{claim}[theorem]{Claim}

\newtheorem*{rep@theorem}{\rep@title}
\newcommand{\newreptheorem}[2]{%
\newenvironment{rep#1}[1]{%
 \def\rep@title{#2 \mbox{\ref{##1}}}%
 \begin{rep@theorem}}%
 {\end{rep@theorem}}}
\newreptheorem{theorem}{Theorem}
\newreptheorem{lemma}{Lemma}
\newreptheorem{proposition}{Proposition}
\newreptheorem{claim}{Claim}

\theoremstyle{sldefinition}
\iflong
\newtheorem{definition}[theorem]{Definition}

\fi

\iflong
\newtheorem{remark}[theorem]{Remark}
\fi

\numberwithin{equation}{section}

\newcommand{\Rbb}{\mathbb{R}}
\newcommand{\Nbb}{\mathbb{N}}

\newcommand{\eps}{\varepsilon}
\newcommand{\set}[1]{\left\{#1\right\}}
\newcommand{\Prb}[2]{\mathrm{Pr}_{#1}\left[#2\right]}
\newcommand{\Exp}[2]{\mathrm{E}_{#1}\left[#2\right]}

\newcommand{\Bin}{\mathrm{Bin}}

\newcommand{\veps}{\varepsilon}
\newcommand{\EX}{\operatornamewithlimits{\mathbb{E}}}
\newcommand{\defeq}{\stackrel{\mathrm{\scriptscriptstyle def}}{=}}

\newcommand{\sgn}{\mathrm{sgn}}
\newcommand{\parhead}[1]{\medskip \noindent {\bfseries\boldmath\ignorespaces #1}\hskip 0.9em plus 0.3em minus 0.3em}

\newcommand{\namedref}[2]{\mbox{\hyperref[#2]{#1~\ref*{#2}}}}

\newcommand{\sectionref}[1]{\namedref{Section}{#1}}
\newcommand{\appendixref}[1]{\namedref{Appendix}{#1}}
\newcommand{\theoremref}[1]{\namedref{Theorem}{#1}}

\newcommand{\definitionref}[1]{\namedref{Definition}{#1}}
\newcommand{\figureref}[1]{\namedref{Figure}{#1}}
\newcommand{\figurerefb}[2]{\hyperref[#1]{Figure~\ref*{#1}#2}}

\newcommand{\lemmaref}[1]{\namedref{Lemma}{#1}}
\newcommand{\claimref}[1]{\namedref{Claim}{#1}}

\newcommand{\equationref}[1]{\mbox{\hyperref[#1]{(\ref*{#1})}}}

\renewcommand{\eqref}{\equationref}

\title{Restricted Isometry Property for General p-Norms%
\thanks{An extended abstract of this paper is to appear at the 31st International Symposium on Computational Geometry - SoCG 2015.}}
\iflong
\author{
Zeyuan Allen-Zhu\thanks{MIT CSAIL, email: \texttt{zeyuan@csail.mit.edu}, \texttt{\string{gelash,ilyaraz\string}@mit.edu}} \\
\and
Rati Gelashvili\footnotemark[2] \\
\and
Ilya Razenshteyn\footnotemark[2] \\
}
\else
\author[1]{Zeyuan Allen-Zhu}
\author[2]{Rati Gelashvili}
\author[3]{Ilya Razenshteyn}
\affil[1]{MIT, \texttt{\href{mailto:zeyuan@csail.mit.edu}{\color{black}zeyuan@csail.mit.edu}}}
\affil[2]{MIT, \texttt{\href{mailto:gelash@mit.edu}{\color{black}gelash@mit.edu}}}
\affil[3]{MIT, \texttt{\href{mailto:ilyaraz@mit.edu}{\color{black}ilyaraz@mit.edu}}}

\authorrunning{Z.\,Allen-Zhu, R.\,Gelashvili, and I.\,Razenshteyn} 

\Copyright{Zeyuan Allen-Zhu, Rati Gelashvili, Ilya Razenshteyn}

\keywords{compressive sensing, dimension reduction, linear algebra, high-dimensional geometry}

\serieslogo{}
\volumeinfo
  {Billy Editor and Bill Editors}
  {2}
  {Conference title on which this volume is based on}
  {1}
  {1}
  {1}
\EventShortName{}
\DOI{10.4230/LIPIcs.xxx.yyy.p}
\fi

\begin{document}

    \maketitle
    \begin{abstract}
    The Restricted Isometry Property (RIP) is a fundamental
    property of a matrix which enables sparse recovery. Informally,
    an $m \times n$~matrix satisfies RIP of order~$k$ for the
    $\ell_p$~norm, if $\|Ax\|_p \approx \|x\|_p$ for
    every vector $x$ with at most~$k$ non-zero coordinates.

    For every $1 \leq p < \infty$
    we obtain almost tight bounds on the minimum number of
    rows~$m$ necessary for the RIP property to hold.
    Prior to this work, only the cases $p = 1$,~$1 + 1 / \log k$, and~$2$
    were studied. Interestingly, our results show that the
    case~$p = 2$ is a ``singularity'' point: the optimal number of rows $m$ is $\widetilde{\Theta}(k^{p})$ for all $p\in [1,\infty)\setminus \{2\}$, as opposed to $\widetilde{\Theta}(k)$ for $k=2$.

    We also obtain almost tight bounds for the column sparsity
    of RIP matrices and discuss implications of our results
    for the Stable Sparse Recovery problem.

    \end{abstract}

    \section{Introduction}

The main object of our interest is a matrix
with \emph{Restricted Isometry Property for the $\ell_p$ norm} (RIP-$p$).
Informally speaking, we are interested in
a linear map from $\Rbb^n$ to $\Rbb^m$ with $m \ll n$
that approximately preserves $\ell_p$ norms for \emph{all} vectors that have only few non-zero coordinates.

More precisely, an $m \times n$ matrix $A \in \Rbb^{m \times n}$ is said to have \emph{$(k, D)$-RIP-$p$ property} for sparsity $k \in [n] \defeq \{1,\dots,n\}$, distortion $D > 1$, and the $\ell_p$ norm for $p \in [1,\infty)$,
if for every vector $x \in \Rbb^n$ with at most $k$ non-zero coordinates
one has
$$
\|x\|_p \leq \|Ax\|_p \leq D \cdot \|x\|_p \enspace.
$$

\noindent
In this work we investigate the following question:  given $p \in [1,\infty)$, $n \in \Nbb$, $k \in [n]$, and $D > 1$,
\begin{center}
\emph{What is the smallest $m \in \Nbb$ so that there exists a
$(k, D)$-RIP-$p$ matrix $A \in \Rbb^{m\times n}$?}
\end{center}
Besides that, the following question arises naturally from the complexity of computing $Ax$:
\begin{center}
\emph{
What is the smallest column sparsity $d$ for such a $(k, D)$-RIP-$p$ matrix $A \in \Rbb^{m\times n}$?
}
\end{center}
(Above, we denote by column sparsity the maximum number of non-zero entries in a column of $A$.)
\subsection{Motivation}

\parhead{Why are RIP matrices important?}
RIP-$2$ matrices were introduced by Cand\`{e}s and Tao~\cite{ct-dlp-05} for
decoding a vector $f$ from corrupted linear measurements $Bf + e$
under the assumption that the vector of errors $e$ is sufficiently sparse
(has only few non-zero entries).
Later Cand\`{e}s, Romberg and Tao~\cite{crt-ssrii-06} used RIP-$2$ matrices to solve \emph{the (Noisy) Stable Sparse Recovery} problem, which has since found numerous applications in areas such as compressive sensing of signals~\cite{crt-ssrii-06,d-cs-06},
genetic data analysis~\cite{kbgsw-pspms-10}, and data stream algorithms~\cite{m-dsaa-05,gi-srusm-10}.

The (noisy) stable sparse recovery problem is defined as follows.
The input signal $x \in \Rbb^n$ is assumed to be close to $k$-sparse, that is, to have most of the ``mass'' concentrated on $k$ coordinates.
The goal is to design a set of $m$ linear measurements that can be represented as a single $m \times n$
matrix $A$ such that, given a \emph{noisy sketch} $y = Ax + e\in \Rbb^m$, where $e\in \Rbb^n$ is a noise vector, one can ``approximately'' recover $x$. Formally, the recovered vector $\widehat{x} \in \Rbb^n$ is required to satisfy
    \begin{equation}
        \label{lp_lq}
        \|x - \widehat{x}\|_p \leq C_1 \min_{\text{$k$-sparse $x^*$}} \|x - x^*\|_1 + C_2 \cdot \|e\|_p
    \end{equation}
    for some $C_1, C_2 > 0$, $p \in[1,\infty)$, and $k \in [n]$.

    (In order for~\eqref{lp_lq} to be meaningful, we also require $\|A\|_p \leq 1$ ---or equivalently, $\|Ax\|_p \leq \|x\|_p$
    for all $x$--- since otherwise, by scaling $A$ up, the noise vector $e$ will become negligible.)

    We refer to~\eqref{lp_lq} as the \emph{$\ell_p/\ell_1$ guarantee}.
    The parameters of interest include: the number of measurements $m$,
    the column sparsity of the measurement matrix $A$, the approximation factors $C_1$, $C_2$ and the complexity of the recovery procedure.

Cand\`{e}s, Romberg and Tao~\cite{crt-ssrii-06} proved that if $A$ is $(O(k), 1 + \eps)$-RIP-$2$ for a sufficiently
small $\eps > 0$, then one can achieve the $\ell_2/\ell_1$ guarantee with $C_1 = O(k^{-1/2})$ and $C_2 = O(1)$
in polynomial time.

The $p=1$ case was first studied by Berinde \emph{et~al.}~\cite{bgiks-cgcua-08}.
They prove that
if $A$ is $(O(k), 1 + \eps)$-RIP-$1$ for a sufficiently small
$\eps > 0$ and has a certain additional property, then one can achieve
the
$\ell_1 / \ell_1$ guarantee with $C_1 = O(1)$, $C_2 = O(1)$.

We note that \emph{any} matrix $A$ that allows the (noisy) stable sparse recovery with the $\ell_p / \ell_1$ guarantee \emph{must have the $(k, C_2)$-RIP-$p$ property}. For the proof see \appendixref{appendix_a}.

\parhead{Known constructions and limitations.}
Cand\`{e}s and Tao~\cite{ct-dlp-05} proved that for every $\eps > 0$, a matrix with $m = O(k \log(n / k) / \eps^2)$ rows and $n$ columns whose entries are sampled from i.i.d. Gaussians
 is $(k, 1 + \eps)$-RIP-$2$ with high probability.
Later, a simpler proof of the same result was discovered by Baraniuk \emph{et al.}~\cite{bddw-sprip-08}%
\footnote{This proof has an advantage that it works for any subgaussian random variables,
    such as random $\pm 1$'s.}.
Berinde \emph{et~al.}~\cite{bgiks-cgcua-08} showed that a (scaled) \emph{random sparse binary matrix} with $m = O(k \log(n / k) / \eps^2)$ rows is $(k, 1 + \eps)$-RIP-$1$ with high probability\footnote{In the same paper~\cite{bgiks-cgcua-08} it is observed that the same construction works for $p = 1 + 1 / \log k$.}. 

Since the number of measurements is very important in practice, it is natural to ask, how optimal is the dimension bound $m = O(k \log(n / k))$ that the above constructions achieve?
The results of Do~Ba~\emph{et~al.}~\cite{dipw-lbsr-10} and Cand\'{e}s~\cite{c-ripii-08}
imply the lower bound $m = \Omega(k \log(n /k))$ for $(k, 1 + \eps)$-RIP-$p$ matrices for $p \in \set{1, 2}$, provided that $\eps > 0$ is sufficiently small.

Another important parameter of a measurement matrix $A$ is its~\emph{column sparsity}: the maximum number of non-zero entries in a single column of $A$. If $A$ has column sparsity $d$, then we can perform multiplication $x \mapsto Ax$ in time
$O(n d)$ as opposed to the naive $O(n m)$ bound. Moreover, for sparse matrices $A$, one can maintain the sketch $y = Ax$ very efficiently if we update $x$. Namely, if we set $x \gets x + \alpha \cdot e_i$, where $\alpha \in \Rbb$ and $e_i \in \Rbb^n$ is a basis vector, then we can update $y$ in time $O(d)$ instead of the naive bound $O(m)$.

The aforementioned constructions of RIP matrices exhibit very different behavior with respect to column sparsity.
RIP-$2$ matrices obtained from random Gaussian matrices are obviously dense, whereas the construction of
RIP-$1$ matrices of Berinde \emph{et al.}~\cite{bgiks-cgcua-08} gives very small column sparsity $d=O(\log(n / k) / \eps)$.
It is known that in both cases the bounds on column sparsity are essentially tight.

Indeed, Nelson and Nguy$\tilde{\hat{\mbox{e}}}$n showed~\cite{nn-slbdr-13} that any
non-trivial column sparsity is impossible for RIP-$2$ matrices unless $m$ is much larger than $O(k \log(n / k))$.
Nachin showed~\cite{nachin} that any RIP-$1$ matrix with
$O(k \log(n / k))$ rows must have column sparsity $\Omega(\log(n / k))$.
Besides that, Indyk and Razenshteyn showed~\cite{ir-mbrm-13}
that every RIP-$1$ matrix
`must be sparse': any RIP-$1$ matrix with $O(k \log(n / k))$ rows
can be perturbed slightly and made $O(\log(n / k))$-sparse.

Another notable difference between RIP-$1$ and RIP-$2$ matrices is the following. The construction of
Berinde \emph{et al.}~\cite{bgiks-cgcua-08} provides RIP-$1$ matrices with non-negative entries, whereas Chandar proved~\cite{c-sgccs-10} that any RIP-$2$ matrix with non-negative entries must have $m = \Omega(k^2)$ (and this was later improved to $m=\Omega(k^2 \log (n/k))$~\cite{nn-slbdr-13,AGMS2014}). In other words, negative signs are crucial in the construction of RIP-$2$ matrices but not for the RIP-$1$ case.

\subsection{Our results}

Motivated by these discrepancies between the optimal constructions for RIP-$p$ matrices with $p \in \big\{1, 1 + \frac{1}{\log k}, 2\big\}$, we initiate the study of RIP-$p$ matrices for the general $p \in [1,\infty)$.

Having in mind that the upper bound $m = O(k \log(n / k))$ holds for RIP-$p$ matrices with $p \in \big\{1, 1 + \frac{1}{\log k}, 2\big\}$,
it would be natural to conjecture that the same bound holds at least
for every $p \in (1,2)$.
As we will see, surprisingly, this conjecture is very far from being true.

Also, knowing that the column sparsity $d = O(k \log (n/k))$ can be obtained for $p=2$ while $d = O(\log (n/k))$ can be obtained for $p=1$, it is interesting to ``interpolate'' these two bounds.

Besides the mathematical interest, a more ``applied'' reason to study RIP-$p$ matrices for the general $p$ is
to get new guarantees for the stable sparse recovery. Indeed, we obtain new results in this direction.

\begin{table}
\begin{center}
\renewcommand{\arraystretch}{1.1}
\begin{tabular}[c]{| >{\centering}m{1.5cm} |c|c| >{\arraybackslash}l |}
    \hline
    $p$ & rows $m$ & column sparsity $d$ & references \\
    \hline
    $1$ & $\Theta(k \log(n / k))$ & $\Theta(\log(n / k))$ &
    \cite{bgiks-cgcua-08,dipw-lbsr-10,nachin,ir-mbrm-13} \\[1mm]
    $1 + \frac{1}{\log k}$ & $O(k \log(n / k))$ & $O(\log(n / k))$ &
    \cite{bgiks-cgcua-08} \\[1mm]
    $(1, 2)$ & $\widetilde{\Theta}(k^p)$ &
    $\widetilde{\Theta}(k^{p-1})$& this work \\[1ex]
    $2$ & $\Theta(k \log(n / k))$ & $\Theta(k \log (n / k))$ &
    \begin{minipage}{6cm}
    \cite{ct-dlp-05,crt-ssrii-06,c-ripii-08,bddw-sprip-08}, \\
    \cite{dipw-lbsr-10,c-sgccs-10,nn-slbdr-13}
    \end{minipage} \\[2ex]
    $(2, \infty)$ & $\widetilde{\Theta}(k^p)$ & $\widetilde{\Theta}(k^{p-1})$ & this work\\
    \hline
\end{tabular}
\end{center}
\caption{Prior and new bounds on RIP-$p$ matrices}
\end{table}

\parhead{Our Upper Bounds.}
On the positive side, for all $\eps > 0$ and all $p \in (1,\infty)$, we construct $(k, 1 + \eps)$-RIP-$p$ matrices with $m = \widetilde{O}(k^p)$ rows. Here, we use the $\widetilde{O}(\cdot)$-notation to hide factors that depend on $\eps$, $p$, and are polynomial in $\log n$.
More precisely, we show that a (scaled) \emph{random sparse
$0/1$ matrix} with $\widetilde{O}(k^p)$ rows and column sparsity $\widetilde{O}(k^{p-1})$ has the desired RIP property with high probability.

This construction essentially matches that of Berinde \emph{et al.}~\cite{bgiks-cgcua-08} when $p$ approaches $1$.
At the same time, when $p = 2$, our result matches known constructions of non-negative RIP-$2$ matrices based on the incoherence argument.%
\footnote{That is, a (scaled) random $m \times n$ binary matrix with $m=O(\veps^{-2}k^2\log(n/k))$ rows and sparsity $d=O(\veps^{-1}k\log(n/k))$ satisfies the $(k,1+\veps)$-RIP-2 property. This can be proved using for instance the incoherence argument from \cite{rauhut2010compressive}: any incoherent matrix satisfies the RIP-2 property with certain parameters.
}

\parhead{Our Lower bounds.}
Surprisingly, we show that, despite our upper bounds being suboptimal for $p = 2$,
the are essentially tight for every constant $p \in (1,\infty)$ except $2$.
Namely, they are optimal both in terms of the dimension $m$ and the column sparsity $d$.

More formally, on the dimension side, for every $p \in (1,\infty) \setminus \{2\}$, distortion $D > 1$, and $(k, D)$-RIP-$p$ matrix $A \in \mathbb{R}^{m\times n}$, we show that $m = \Omega(k^p)$, where $\Omega(\cdot)$
hides factors that depend on $p$ and $D$. Note that, it is not hard to extend an argument of Chandar~\cite{c-sgccs-10} and obtain a lower bound $m = \Omega(k^{p-1})$.%
\footnote{Also, the same argument gives the lower bound $\Omega(k^p)$ for \emph{binary} RIP-$p$ matrices for every $p \in [1, \infty)$.}
This additional factor $k$ is exactly what makes our lower bound non-trivial and tight for $p\in (1,\infty)\setminus \{2\}$, and thus enables us to conclude that $p = 2$ is a ``singularity''.%
\footnote{A similar singularity is known to exist for
linear dimension reduction for arbitrary point sets with respect to $\ell_p$~norms~\cite{lmn-msdsa-05}; alas, tight bounds
for that problem are not known.}

As for the column sparsity, we present a simple extension of the argument of Chandar~\cite{c-sgccs-10}
and prove that for every $p \in [1,\infty)$ any $(k, D)$-RIP-$p$ matrix must have column sparsity
$\Omega(k^{p-1})$.


\parhead{RIP matrices and sparse recovery.}
We extend the result of Cand\`{e}s, Romberg and Tao~\cite{crt-ssrii-06} to show that, for every $p > 1$, RIP-$p$ matrices allow the stable sparse recovery with the $\ell_p / \ell_1$
guarantee and approximation factors
$C_1 = O\bigl(k^{-1 + 1/p}\bigr)$, $C_2 = O(1)$ in polynomial time. This extension is quite straightforward and seems to be folklore, but, to the best of our knowledge, it is not recorded anywhere.

On the other hand, for every $p\geq 1$, it is almost immediate that \emph{any} matrix $A$ that allows the stable sparse
recovery with the $\ell_p / \ell_1$ guarantee ---even if it works only for $k$-sparse signals--- \emph{must have the $(k, C_2)$-RIP-$p$ property}.
For the sake of completeness, we have included both the above proofs in \appendixref{appendix_a}.

\parhead{Implications to sparse recovery.}
Using the above equivalent relationship between the stable sparse recovery problem and the RIP-$p$ matrices, we conclude that the stable sparse recovery with the $\ell_p / \ell_1$ guarantee requires $m=\widetilde{\Theta}(k^p)$ measurements for every $p \in [1; \infty) \setminus \set{2}$, and requires $d=\tilde{\Theta}(k^{p-1})$ column sparsity for every $p \in [1, \infty)$. Our results together draw tradeoffs between the following three parameters in stable sparse recovery:
\begin{itemize}[nolistsep]
\item $p$, the $\ell_p/\ell_1$ guarantee for the stable sparse recovery,%
\footnote{We note that the $\ell_p/\ell_1$ and the $\ell_q/\ell_1$ guarantees are incomparable. However, it is often more desirable to have larger $p$ in this $\ell_p/\ell_1$ guarantee to ensure a better recovery quality. This is because, if the noise vector $e=0$, the $\ell_q/\ell_1$ guarantee (with $C_1=O(k^{-1+1/q})$) can be shown to be stronger than the $\ell_p/\ell_1$ one (with $C_1=O(k^{-1+1/p})$) whenever $q>p$. However, when there is a noise term, the guarantee $\|x-\hat{x}\|_p \leq O(1) \cdot \|e\|_p$ is incomparable to $\|x-\hat{x}\|_q \leq O(1) \cdot \|e\|_q$ for $p\neq q$.}
\item $m$, the number of measurements needed for sketching, and
\item $d$, the running time (per input coordinate) needed for sketching.
\end{itemize}

\medskip
It was pointed out by an anonymous referee that for the \emph{noiseless} case ---that is, when the noise vector
$e$ is always zero--- better
upper bounds are possible. Using the result of Gilbert \emph{et al.}~\cite{gstv-osfaf-07}, one can obtain, for every
$p \geq 2$, the noiseless stable sparse recovery procedure with the $\ell_p / \ell_1$ guarantee using only
$m=\widetilde{O}(k^{2 - 2/p})$ measurements.
Therefore, our results also imply a very large gap, both in terms of $m$ and $d$, between the \emph{noiseless} and the \emph{noisy} stable sparse recovery problems.

\subsection{Overview of the proofs}

\parhead{Upper bounds.}
We construct RIP-$p$ matrices as follows.
Beginning with a zero matrix $A$ with $m=\widetilde{O}(k^p)$ rows and $n$ columns, independently for each column of $A$, we choose $d = \widetilde{O}(k^{p-1})$ out of $m$ entries uniformly at random (without replacement), and assign the value $d^{-1/p}$ to those selected entries.
For this construction, we have two very different analyses of its correctness: one works only for $p \geq 2$, and the other works only for $1 < p < 2$.

For $p \geq 2$, the most challenging part is to show that $\|Ax\|_p \leq (1+\veps)\|x\|_p$ holds with high probability, for all $k$-sparse vectors $x$. We reduce this problem to a probabilistic question
\emph{similar in spirit} to the following ``balls and bins'' question.
Consider $n$ bins in which we throw $n$ balls uniformly and independently.
As a result, we get $n$ numbers $X_1$, $X_2$, \ldots, $X_n$, where $X_i$ is the number of balls falling into the
$i$-th bin.
We would like to upper bound the tail $\Prb{}{S \geq 1000 \cdot \Exp{}{S}}$ for the random variable $S = \sum_{i=1}^n X_i^{p-1}$. (Here, the constant $1000$ can be replaced with any large enough one since we do not care about constant factors in this paper.)
The first challenge is that $X_i$'s are not independent.
To deal with this issue we employ
the notion of \emph{negative association} of random variables
introduced by Joag-Dev and Proschan~\cite{jp-narva-83}.
The second problem is that the random variables $X_i^{p-1}$
are heavy tailed: they have tails of the form
$\Pr_{}\big[X_i^{p-1} \geq t \big] \approx \exp(-t^{\frac{1}{p - 1}})$,
so the standard technique of bounding the moment-generating function
does not work.
Instead, we bound the high moments of~$S$ directly,
which introduces certain technical challenges.
Let us remark that sums of i.i.d. heavy-tailed variables were thoroughly studied by Nagaev~\cite{n-iltt1-69,
n-iltt2-69}, but it seems that for the results in these papers the independence of summands is crucial.

One major reason the above approach fails to work for $1 < p < 2$ is that, in this range,
even the best possible tail inequality for $S$ is too weak for our purposes.
Another challenge in this regime is that, to bound the ``lower tail'' of $\|Ax\|_p^p$ (that is, to prove that
$\|Ax\|_p \geq (1 - \eps) \|x\|_p$ holds for all $k$-sparse $x$), the simple argument used for $p\geq 2$ no longer
works.
Our solution to both problems above is to instead build our RIP matrices based on the following general notion of bipartite expanders.

\begin{definition}\label{def:bipartite-expander}
    Let $G = (U, V, E)$ with $|U| = n$, $|V| = m$ and $E \subseteq U \times V$ be a bipartite graph
    such that all vertices from $U$ have the same degree $d$. We say that $G$ is an
    \emph{$(\ell, d, \delta)$-expander}, if for every $S \subseteq U$ with $|S| \leq \ell$ we have
    $$
        \big|\set{v \in V \mid \exists u \in S \enspace (u, v) \in E}\big| \geq (1 - \delta) d |S| \enspace.
    $$
\end{definition}

\noindent
It is known that random $d$-regular graphs are good expanders, and we can take the (scaled) adjacency matrix of such an expander and prove that it satisfies the desired RIP-$p$ property for $1<p<2$.
Our argument can be seen as a subtle interpolation
between the argument from~\cite{bgiks-cgcua-08},
which proves that (scaled) adjacency matrices of $(k, d, \Theta(\eps))$-expanders
(with $\widetilde{O}(k)$ rows)
are $(k, 1 + \eps)$-RIP-$1$ and the one
using incoherence argument,%
\footnote{It is known \cite{rauhut2010compressive} that an incoherent matrix satisfies the RIP-2 property with certain parameters. At the same time, the notion of incoherence can be interpreted as expansion for $\ell=2$.}
 which shows that $(2, d, \Theta(\eps/k))$-expanders
give $(k, 1 + \eps)$-RIP-$2$ matrices (with $\widetilde{O}(k^2)$ rows).

\parhead{Lower bounds.}
Our dimension lower bound $m = \Omega(k^{p})$ is derived essentially from norm inequalities. The high-level idea can be described in four simple steps. Consider any $(k,D)$-RIP-$p$ matrix $A\in\Rbb^{n\times m}$, and assume that $D$ is very close to $1$ in this high-level description.

In the first three steps, we deduce from the RIP property that (a) the sum of the $p$-th powers of all entries in $A$ is approximately $n$, (b) the largest entry in $A$ (i.e., the vector $\ell_\infty$-norm of $A$) is essentially at most $k^{1/p-1}$, and (c) the sum of squares of all entries in $A$ is at least $n \big(\frac{k}{m}\big)^{2/p-1}$ if $p\in (1,2)$, or at most $n \big(\frac{k}{m}\big)^{2/p-1}$ if $p>2$. In the fourth step, we combine (a) (b) and (c) together by arguing about the relationships between the $\ell_p$, $\ell_\infty$ and $\ell_2$ norms of entries of $A$,
and prove the desired lower bound on $m$.

The sparsity lower bound $d = \Omega(k^{p-1})$ can be obtained via a simple extension of the argument of
Chandar~\cite{c-sgccs-10}.
It is possible to extend the techniques of Nelson and Nguy$\tilde{\hat{\mbox{e}}}$n~\cite{nn-slbdr-13} to obtain
a slightly better sparsity lower bound. However, since we were unable to obtain \emph{a tight} bound this way, we decided
not to include it.

\section{RIP Construction for $p \geq 2$}
\label{sec:positive-p-larger-2}
In this section, we construct $(k,1+\veps)$-RIP-$p$ matrices for $p\geq 2$ by proving the following theorem.
\begin{definition}
\label{def:random-matrix}
We say that an $m\times n$ matrix $A$ is a \emph{random binary matrix with sparsity $d\in [m]$}, if $A$ is generated by assigning $d^{-1/p}$ to $d$ random entries per column (selected uniformly at random without replacement), and assigning $0$ to the remaining entries.
\end{definition}
\begin{theorem}
    \label{rip_larger_2}
    For all $n\in \mathbb{Z}_+$, $k \in [n]$, $\eps\in (0,\frac{1}{2})$ and $p \in [2,\infty)$, there exist $m,d \in \mathbb{Z}_+$ with
    $$
        m = p^{O(p)} \cdot \frac{k^p}{\eps^2} \cdot \log^{p-1} n
        \mbox{\quad and \quad}
        d = p^{O(p)} \cdot \frac{k^{p-1}}{\eps} \cdot \log^{p-1} n \leq m
    $$
such that, letting $A$ be a random binary $m \times n$ matrix of sparsity $d$, with probability at least $98\%$, $A$ satisfies
$(1-\veps)\|x\|_p^p \leq \|Ax\|_p^p \leq (1+\veps)\|x\|_p^p$
for all $k$-sparse vectors $x\in \mathbb{R}^n$.
\end{theorem}

Our proof is divided into two steps: (1) the ``lower-tail step'', that is, with probability at least $0.99$ we have $\|Ax\|_p^p \geq (1-\veps)\|x\|_p^p$ for all $k$-sparse $x$, and (2) the ``upper-tail step'', that is, with probability at least $0.99$, we have $\|Ax\|_p^p \leq (1+\veps)\|x\|_p^p$.

For every $j \in [n]$, let us denote by $S_j \subseteq [m]$ the set of non-zero rows of the $j$-th column of $A$.
\subsection{The Lower-Tail Step}
To lower-tail step is very simple. It suffices to show that, with high probability, $|S_i \cap S_j|$ is small for every pair of different $i,j \in [n]$, which will then imply that if only $k$ columns of $A$ are considered, every $S_i$ has to be almost disjoint from the union of the $S_j$ of the $k-1$ remaining columns. This can be summarized by the following claim.

    \begin{claim}
    \label{claim:incoherent}
        If $d \geq C \veps^{-1} k \log n $ and $m \geq 2dk / \eps$, where $C$ is some large enough constant, then
        $$
            \Pr \Big[ \forall 1 \leq i < j \leq n \quad |S_i \cap S_j| \leq \frac{\eps d}{k} \Big] \geq 0.99 \enspace.
        $$
    \end{claim}
    \begin{proof}
     Let us first upper bound the probability that $S_i$ and $S_j$ intersect by more than $\frac{\veps d}{k}$
     elements. For notational simplicity suppose that $S_i = \{1,\dots,d\}$, and let the random variable $X_k$ be $1$ if $S_j$ contains $k$, and $0$ if not.
     Under this definition, we have $|S_i \cap S_j|=\sum_{i=1}^{d} X_i$.

     Noticing that the expectation $\EX[X_1+\cdots+X_d]=\frac{d}{m} \cdot d = \frac{d^2}{m}$, and $\frac{\veps d}{k} \geq 2\cdot \frac{d^2}{m}$ is twice as large as the expectation, we apply Chernoff bound
     for negatively correlated binary random variables~\cite{ps-rdece-97}
     and obtain
     $$\Pr \Big[|S_i\cap S_j|>\frac{\veps d}{k} \Big] = \Pr \Big[X_1+\cdots+X_d > \frac{\veps d}{k}\Big] < e^{-\Omega(\veps d/k)} \leq \frac{1}{100 n^2} \enspace,$$
     where the last inequality is true by our choice of $d \geq C \veps^{-1} k \log n$ for some large enough constant $C$. Finally, by union bound, we have $\Pr \big[ \exists i,j\in [n] \mbox{ with $i \ne j$}, \, |S_i\cap S_j| > \frac{\veps d}{k} \big]\leq 0.01$.
    \end{proof}

Now, to prove the lower tail, without loss of generality, let us assume that $x$ is supported on $[k]$, the first $k$ coordinates. For every $j \in [k]$, we denote by $S_j' = S_j \setminus \bigcup_{j' \in [k]\setminus \{j\}} S_{j'}$, the set of non-zero rows in column $j$ that are not shared with the supports of
other columns in $[k]\setminus \{j\}$. If the event in \claimref{claim:incoherent} holds, then for every $j \in [k]$, we have $|S_j'| \geq (1 - \eps) d$. Thus, we can lower bound $\|A x\|_p$ as
\begin{equation}\label{eqn:p-more-than-2:lower-tail}
        \|Ax\|_p^p =
\frac{1}{d} \cdot \sum_{i=1}^m \bigg|\sum_{{j \in [k]: i \in S_j}} x_j\bigg|^p
\geq \frac{1}{d} \cdot \sum_{i=1}^m \bigg|\sum_{{j \in [k]:  i \in S_j'}} x_j\bigg|^p
        =
        \frac{1}{d} \cdot \sum_{j \in [k]} |S_j'| \cdot |x_j|^p \geq (1 - \eps) \|x\|_p^p \enspace.
\end{equation}

\begin{remark}
The above claim only works when $m=\Omega(k^2 \log n/\veps^2)$, and therefore we cannot use it in for the case of $1<p<2$.
\end{remark}

\subsection{The Upper-Tail Step}
Below we describe the framework of our proof for the upper-tail step, deferring all technical details to \sectionref{app:positive-p-larger-2}.

Suppose again that $x$ is supported on $[k]$. Then, we upper bound $\|A x\|_p^p$ as
\begin{align}
    \|Ax\|_p^p =
    \frac{1}{d} \cdot \sum_{i=1}^m \bigg|\sum_{j \in [k]: i \in S_j} x_j\bigg|^p
&\leq
\frac{1}{d} \cdot \sum_{i=1}^m \big|\set{j' \in [k] \mid i \in S_{j'}}\big|^{p-1}
    \cdot \sum_{j \in [k]: i \in S_j } |x_j|^p \nonumber \\
    &= \frac{1}{d} \cdot \textstyle{\sum}_{j=1}^k |x_j|^p \cdot \textstyle{\sum}_{i \in S_j} \big|\set{j' \in [k] \mid i \in S_{j'}}\big|^{p-1} \enspace, \label{eqn:p-more-than-2:upper-tail}
\end{align}
where the first inequality follows from the fact that $(a_1+\cdots+a_N)^p \leq N^{p-1} (a_1^p + \cdots + a_N^p)$ for any sequence of $N$ non-negative reals $a_1,\dots,a_N$. Note that the quantity $\big|\set{j' \in [k] \mid i \in S_{j'}}\big| \in [k]$ captures the number of non-zeros of $A$ in the $i$-th row and the first $k$ columns. From now on, in order to prove the desired upper tail, it suffices to show that, with high probability
\begin{equation}\label{eqn:binomial-goal}
  \forall j\in [k],\quad    \textstyle{\sum}_{i \in S_j} \big|\set{j' \in [k] \mid i \in S_{j'}}\big|^{p-1} \leq (1 + \eps) d \enspace.
\end{equation}

To prove this, let us fix some $j^* \in [k]$ and upper bound the probability that \equationref{eqn:binomial-goal} holds for $j=j^*$, and then take a union bound over the choices of $j^*$. Without loss of generality, assume that $S_{j^*}=\{1,2,\dots,d\}$, consisting of the first $d$ rows. For every $i \in S_{j^*}$, define
a random variable $X_i \defeq \big|\set{j' \in [k] \mid i \in S_{j'}}\big| - 1$. It is easy to see that $X_i$ is distributed as $\Bin(k - 1, d / m)$, the binomial distribution that is the sum of $k-1$ i.i.d. random $0/1$
variables, each being $1$ with probability $d/m$.
For notational simplicity, let us define $\delta \defeq dk/m$. We will later choose $\delta < \veps$ to be very small. Our goal in \equationref{eqn:binomial-goal} can now be reformulated as follows: upper bound the probability
$$ \Pr \Big[ \; \textstyle{\sum}_{i=1}^d ((X_i+1)^{p-1} - 1) > \veps d \; \Big] \enspace. $$

We begin with a lemma showing an upper bound on the moments of each $Y_i \defeq (X_i+1)^{p-1} - 1$.
\begin{lemma}
\label{lemma:p-more-than-2:single-moment}
There exists a constant $C \geq 1$ such that, if $X$ is drawn from the binomial distribution $\Bin(k-1,\delta/k)$ for some $\delta < 1/(2e^{2})$, and $p \geq 2$, then for any real $\ell \geq 1$,
$$ \EX[((X+1)^{p-1}-1)^\ell] \leq  C \cdot \delta (\ell(p-1)+1)^{\ell(p-1)+1} \enspace. $$
\end{lemma}

Next, we note that although the random variables $X_i$'s are dependent, they can be verified to be \emph{negatively associated}, a notion introduced by Joag-Dev and Proschan~\cite{jp-narva-83}. This theory allows us to conclude the following bound on the moments.

\begin{lemma}
\label{lemma:negative-asso}
Let $\widetilde{X}_1,\dots,\widetilde{X}_d$ be $d$ random variables, each drawn \emph{independently} from $\Bin(k - 1, \delta/k)$. Then, for every integer $t \geq 1$ we have
$$
    \EX \left[\left(\textstyle{\sum}_{i=1}^d ((X_i + 1)^{p-1} - 1)\right)^t\right]
    \leq
    \EX \left[\left(\textstyle{\sum}_{i=1}^d ((\widetilde{X}_i + 1)^{p-1} - 1)\right)^t\right] \enspace.
$$
\end{lemma}

Now, using the moments of random variables $Y_i = (X_i+1)^{p-1}-1$ from \lemmaref{lemma:p-more-than-2:single-moment}, as well as \lemmaref{lemma:negative-asso}, we can compute the tail bound of the sum $\sum_{i=1}^d Y_i$. Our proof of the following Lemma uses the result of Lata{\l}a~\cite{l-emsir-97}.

\begin{lemma}\label{lemma:p-more-than-2:together-moment}
There exists constants $C \geq 1$ such that, whenever $\delta \leq \veps/p^{C p}$ and
$d \geq p^{Cp} / \eps$, we have
$$ \Pr \left[ \textstyle{\sum}_{i=1}^d ((X_i+1)^{p-1} - 1) > \veps d \right] \leq e^{-\Omega(\frac{(\veps d)^{1/(p-1)}}{p})} \enspace.$$
\end{lemma}

\noindent
Finally, we are ready to prove \theoremref{rip_larger_2}.

\begin{proof}[Proof of \theoremref{rip_larger_2}]
We can choose $d = \Theta(p)^{p-1} \cdot \frac{k^{p-1}}{\eps} \cdot \log^{p-1} n$ so that $e^{-\Omega(\frac{(\veps d)^{1/(p-1)}}{p})} < \frac{1}{100} \frac{1}{k\binom{n}{k}}$. Since our choice of $m = \frac{dk p^{\Theta(p)}}{\veps}$ ensures that $\delta = dk/m \leq \veps/p^{Cp}$, and our choice of $d$ ensures $d \geq p^{Cp} / \eps$, we can apply \mbox{\lemmaref{lemma:p-more-than-2:together-moment}}  and conclude that with probability at least $1-\frac{1}{100} \frac{1}{k\binom{n}{k}}$ one has
\begin{equation*}
  \textstyle{\sum}_{i \in S_{j^*}} \big|\set{j' \in [k] \mid i \in S_{j'}}\big|^{p-1} = \textstyle{\sum}_{i=1}^d (X_i+1)^{p-1} \leq (1 + \eps) d \enspace.
\end{equation*}

Therefore, by applying the union bound
over all $j^* \in [k]$, we conclude that with probability at least $1 - \frac{1}{100} \frac{1}{\binom{n}{k}}$, the
desired inequality \equationref{eqn:binomial-goal} is satisfied for all $j\in [k]$.

Recall that, owing to \equationref{eqn:p-more-than-2:upper-tail}, the inequality \equationref{eqn:binomial-goal} implies that $\|Ax\|_p^p \leq (1+\veps)\|x\|_p^p$ for every $x\in \mathbb{R}^n$ that is supported on the \emph{first} $k$ coordinates. By another union bound over the choices of all possible $\binom{n}{k}$ subsets of $[n]$, we conclude that with probability at least $0.99$, we have $\|Ax\|_p^p \leq (1+\veps)\|x\|_p^p$ for all $k$-sparse vectors
$x$.

On the other hand, since our choice of $d$ and $m$ satisfies the assumptions $d \geq \Omega(k\log n/\veps)$ and $m\geq 2dk/\veps$ in \claimref{claim:incoherent}, the lower tail $\|Ax\|_p^p \geq (1-\veps)\|x\|_p^p$ also holds with probability at least $0.99$.  Overall we conclude that with probability at least $0.98$, we have $\|Ax\|_p^p \in (1\pm \veps)\|x\|_p^p$ for every $k$-sparse vector $x \in \mathbb{R}^n$.
\end{proof}

\subsubsection{Missing Proofs}
\label{app:positive-p-larger-2}

\begin{replemma}{lemma:p-more-than-2:single-moment}
There exists some constant $C \geq 1$ such that, if $X$ is drawn from the binomial distribution $\Bin(k-1,\delta/k)$ for some $\delta < 1/(2e^{2})$, and $p \geq 2$, then for any real $\ell \geq 1$,
$$ \EX[((X+1)^{p-1}-1)^\ell] \leq  C \cdot \delta (\ell(p-1)+1)^{\ell(p-1)+1} \enspace. $$
\end{replemma}
\begin{proof}
We first expand the expectation using the definition of $\Bin(k-1,\delta/k)$.
\begin{align*}
\EX[((X+1)^{p-1}-1)^\ell]
&= \sum_{i=0}^{k-1} ((i+1)^{p-1}-1)^\ell \binom{k-1}{i} \left(1-\frac{\delta}{k}\right)^{k-1-i} \left( \frac{\delta}{k} \right)^i \\
&\leq \sum_{i=1}^{k-1} ((i+1)^{p-1}-1)^\ell \Big(\frac{e(k-1)}{i}\Big)^i \left( \frac{\delta}{k} \right)^i \\
&\leq \sum_{i=1}^{\infty} (i+1)^{\ell(p-1)} \left( \frac{e \delta}{i} \right)^i \enspace.
\end{align*}
Let us denote by $a_i \defeq (i+1)^{\ell(p-1)} \big( \frac{e \delta}{i} \big)^i$ the $i$-th term of the above infinite sum.
We have
$$
    \frac{a_i}{a_{i-1}} = \left(1+\frac{1}{i}\right)^{\ell(p-1)}\frac{e \delta}{i}\left(1-\frac{1}{i}\right)^{i-1}
    \leq \delta \cdot e^{\ell(p-1)/i+1} \enspace.
$$
Since $\delta < 1/(2e^{2})$, we have $a_i / a_{i-1} < 1/2$ for every $i \geq \max\{\ell(p-1),2\}$. Therefore, the largest $\max_{i\geq 1} a_i$ is obtained when $i=i^*<\max\{\ell(p-1),2\}$, which implies $1 \leq i^* \leq \ell(p-1)$ because $\ell(p-1)\geq 1$ (here we crucially use that $p \geq 2$ and $\ell \geq 1$).
Therefore,
$$
    \max_{i\geq 1} a_i \leq \Big(\frac{e\delta}{i^*}\Big)^{i^*} \cdot (\ell(p-1)+1)^{\ell(p-1)} \leq e \delta \cdot (\ell(p-1)+1)^{\ell(p-1)}  \enspace,
$$
and the second inequality is because $e\delta < 1$.
Overall,
$$
    \EX[((X+1)^{p-1}-1)^{\ell}] \leq \sum_{i=1}^{\infty} a_i \leq \left(\ell(p-1)+\sum_{j=1}^{\infty}2^{-j}\right) \cdot \max_i a_i
    \leq
    O(\delta) \cdot (\ell(p-1)+1)^{\ell(p-1)+1} \enspace. \qedhere
$$
\end{proof}

\begin{replemma}{lemma:negative-asso}
Letting $\widetilde{X}_1,\dots,\widetilde{X}_d$ be $d$ random variables, each drawn \emph{independently} from $\Bin(k - 1, \delta/k)$. Then, for every integer $t \geq 1$ we have
$$
    \EX \left[\left(\sum_{i=1}^d ((X_i + 1)^{p-1} - 1)\right)^t\right]
    \leq
    \EX \left[\left(\sum_{i=1}^d ((\widetilde{X}_i + 1)^{p-1} - 1)\right)^t\right] \enspace.
$$
\end{replemma}
\begin{proof}
This lemma follows from the theory of \emph{negatively associated} random variables \cite{jp-narva-83} (see also~\cite{e-ippfe-65}).

For every $i\in [m]$ and $j\in [k]$, let the random variable $Z_{ij} = 1$ if $i \in S_j$ and $0$ otherwise. The random variables across columns are independent: that is, $\{Z_{1,j},\dots,Z_{m,j}\}$ and $\{Z_{1,j'},\dots,Z_{m,j'}\}$ are independent if $j \neq j'$.

However, within a single column, $Z_{1j},\dots,Z_{mj}$ are not independent. In fact, this distribution
can be seen as the uniform distribution over $\set{0, 1}^m$ conditioned on $Z_{1j} + \ldots + Z_{mj} = d$. Thus, applying \cite[Theorem 2.8]{jp-narva-83}, we have that $Z_{1j},Z_{2j},\dots,Z_{mj}$ are negatively associated (this uses the fact that the Bernoulli distribution is a \emph{P\'{o}lya frequency function of order two}).
Now, combining the $k$ independent columns, we get
that the variables $\{Z_{ij}\}_{i\in[m],j\in[k]}$ are negatively associated altogether.

Next, we want to show that the variables $\{X_1,\dots,X_d\}$ are also negatively associated. By definition, we have $X_i = \sum_{j \neq j^*} Z_{i,j}$. Therefore, $X_1,\dots,X_d$ is a sequence of random variables, each being a partial sum of $\{Z_{ij}\}_{i\in[m],j\in[k]}$, and different $X_i$'s cover disjoint subsets of $\{Z_{ij}\}_{i\in[m],j\in[k]}$. Applying \cite[Property $P_6$]{jp-narva-83}, we have that the variables $\{X_1,\dots,X_d\}$ are negatively associated.

Finally, since the function $f(x) = (x+1)^{p-1} - 1$ is non-decreasing for $p\geq 1$, we apply \cite[Property $P_6$]{jp-narva-83} and conclude that the variables $\{ (X_i+1)^{p-1} - 1 \}_{i\in [d]}$ are also negatively associated. Letting $Y_i \defeq (X_i+1)^{p-1} - 1$, then \cite[Property $P_2$]{jp-narva-83} gives that $\EX[\prod_{i\in S} Y_i^{r_i}] \leq \prod_{i\in S} \EX[Y_i^{r_i}]$ for any subset $S\subseteq [d]$ and any sequence of powers $r_1,\dots,r_d\in \mathbb{Z}_{\geq 0}$.

As a result, we conclude that, letting $\widetilde{Y}_i \defeq (\widetilde{X}_i+1)^{p-1}-1$,
\begin{align*}
\EX \Big[ \big(\sum_{i=1}^d Y_i \big)^t\Big]
&= \EX\Big[ \sum_{\substack{r_1,\dots,r_d \in \{0,1,\dots,d\} \\ r_1 + \cdots + r_d = t}} \binom{t}{r_1,\dots,r_d} Y_1^{r_1} Y_2^{r_2} \cdots Y_d^{r_d} \Big] \\
&\leq
\sum_{r_1 + \cdots + r_d = t} \binom{t}{r_1,\dots,r_d} \EX[Y_1^{r_1}] \EX[Y_2^{r_2}] \cdots \EX[Y_d^{r_d}] \\
&=
\sum_{r_1 + \cdots + r_d = t} \binom{t}{r_1,\dots,r_d} \EX[\widetilde{Y}_1^{r_1}] \EX[\widetilde{Y}_2^{r_2}] \cdots \EX[\widetilde{Y}_d^{r_d}]
=  \EX \left[\left(\sum_{i=1}^d \widetilde{Y}_i\right)^t\right] \enspace. \tag*{\qedhere}
\end{align*}
\end{proof}

To prove \lemmaref{lemma:p-more-than-2:together-moment}, we need the following theorem of Lata{\l}a on the moments of sums of i.i.d. non-negative random variables.

\begin{theorem}[\cite{l-emsir-97}, Theorem 1]
\label{thm:latala}
Let $Y_1,\dots,Y_d$ be a sequence of independent non-negative random variables from distribution $\mathcal{D}$ and $t\geq 1$. Then
$$ \EX[(Y_1+\dots+Y_d)^t]^{1/t} \leq e \cdot \inf \bigg\{u>0 \,:\, \EX_{Y\sim \mathcal{D}}\Big[\left(1+\frac{Y}{u}\right)^t\Big] \leq e^{t/d} \bigg\} \enspace.$$
\end{theorem}

\begin{replemma}{lemma:p-more-than-2:together-moment}
There exists a constant $C \geq 1$ such that, whenever $\delta \leq \veps/p^{C p}$
and $d \geq p^{Cp} / \eps$, we have
$$ \Pr \left[ \textstyle{\sum}_{i=1}^d ((X_i+1)^{p-1} - 1) > \veps d \right] \leq e^{-\Omega(\frac{(\veps d)^{1/(p-1)}}{p})} \enspace.$$
\end{replemma}
\begin{proof}
Denote by $Y$ the random variable whose value $Y = (X+1)^{p-1}-1$, where $X$ is drawn from the binomial distribution $\Bin(k-1,\delta/k)$, and $\mathcal{D}$ the distribution for $Y$.
We wish to apply \theoremref{thm:latala} for the case of $d$ independent samples from $\mathcal{D}$, and let us compute the value of $u$ from the statement of \theoremref{thm:latala} as follows. For every integer $t\geq 1$,
\begin{align*}
\EX_{Y\sim \mathcal{D}}\Big[\big(1+\frac{Y}{u}\big)^t\Big]
= 1 + \sum_{\ell=1}^t \binom{t}{\ell} \frac{\EX[Y^\ell]}{u^\ell}
\leq 1 + \sum_{\ell=1}^t \big(\frac{et}{\ell}\big)^\ell \frac{\EX[Y^\ell]}{u^\ell} \enspace.
\end{align*}
This sum, owing to \lemmaref{lemma:p-more-than-2:single-moment}, can be upper bounded as
\begin{align*}
\sum_{\ell=1}^{t} \big(\frac{et}{u\ell}\big)^\ell \EX[Y^\ell]
&\leq O(\delta) \cdot \sum_{\ell=1}^{t} \big(\frac{et}{u\ell}\big)^\ell \cdot (\ell(p-1)+1)^{\ell(p-1)+1}
\leq O(\delta) \cdot \sum_{\ell=1}^{t} \big(\frac{et}{u\ell}\big)^\ell \cdot (\ell p)^{\ell(p-1)+1} \\
&\leq O(\delta p) \cdot \sum_{\ell=1}^{t} \ell \big(\frac{et}{u\ell}\big)^\ell \cdot (\ell p)^{\ell(p-1)}
= O(\delta p) \cdot \sum_{\ell=1}^{t} \ell \big(\frac{ep^{p-1} \cdot t \ell^{p-2} }{u}\big)^\ell \\
&= O(\frac{\delta p^p t}{u}) \cdot \sum_{\ell=1}^{t} \ell^{p-1} \big(\frac{ep^{p-1} \cdot t \ell^{p-2} }{u}\big)^{\ell-1}
\leq O(\frac{\delta p^p t}{u}) \cdot \sum_{\ell=1}^{t} \ell^{p-1} \big(\frac{ep^{p-1} \cdot t^{p-1} }{u}\big)^{\ell-1} \enspace.
\end{align*}
Above, the last inequality has used the fact that $p\geq 2$. Now, by choosing $u \geq 2ep^{p-1} \cdot t^{p-1}$ we have that
\begin{align}
\sum_{\ell=1}^{t} \big(\frac{et}{u\ell}\big)^\ell \EX[Y^\ell]
\leq O(\frac{\delta p^p t}{u}) \cdot \sum_{\ell=1}^{t} \ell^{p-1} \frac{1}{2^{\ell-1}} \enspace.
\label{eqn:p-more-than-2:latala-sum}
\end{align}
Since
$$
    \sum_{\ell=1}^{\infty}\ell^{p-1} \frac{1}{2^{\ell-1}}\leq 2 \cdot \int_{0}^{\infty} x^{p-1} \cdot 2^{-x} \; dx
    \leq 2^{O(p)} \cdot \Gamma(p) \leq p^{O(p)} \enspace,
$$
we conclude that the right hand side of \equationref{eqn:p-more-than-2:latala-sum} is upper bounded by $O(\frac{\delta p^{O(p)} t}{u})$. In sum, we conclude that when $u = 2ep^{p-1} \cdot t^{p-1}$ and $t^{p-1} \geq \delta p^{\Omega(p)} d$, we have
\begin{align*}
\EX_{Y\sim \mathcal{D}}\Big[\big(1+\frac{Y}{u}\big)^t\Big]
\leq 1 + O(\frac{\delta p^{O(p)} t}{u}) \leq 1 + \frac{t}{d} < e^{t/d} \enspace.
\end{align*}
Invoking \theoremref{thm:latala} for this choice of $u = 2ep^{p-1} \cdot t^{p-1}$ and for any integer $t\geq 1$ satisfying $t^{p-1} \geq \delta p^{\Omega(p)} d$, we have
$$\EX \left[ \left(\sum_{i=1}^d ((\widetilde{X}_i+1)^{p-1} - 1) \right)^t \right]^{1/t} \leq 2e^2p^{p-1} \cdot t^{p-1} \enspace,$$
where each $\widetilde{X}_i$ is an i.i.d. sample from $\Bin(k-1,\delta/k)$.  Invoking \lemmaref{lemma:negative-asso}, we obtain the same moment bound on $X_1,\dots,X_d$.
$$\EX \left[ \left(\sum_{i=1}^d ((X_i+1)^{p-1} - 1) \right)^t \right]^{1/t} \leq 2e^2p^{p-1} \cdot t^{p-1} \enspace.$$

Using Markov's inequality, we have for any integer $t\geq 1$ satisfying $t^{p-1} \geq \delta p^{\Omega(p)} d$,
$$\Pr \left[ \sum_{i=1}^d ((X_i+1)^{p-1} - 1) > \veps d \right] \leq \frac{1}{(\veps d)^t} \EX \left[ \left(\sum_{i=1}^d ((X_i+1)^{p-1} - 1)\right)^t \right] \leq \left(\frac{2e^2p^{p-1} \cdot t^{p-1}}{\veps d} \right)^t \enspace.$$

By the assumption $d \geq p^{C p} / \veps$, so let us choose $t$ to be the largest positive integer such that
$\frac{2e^2p^{p-1} \cdot t^{p-1}}{\veps d} \leq \frac{1}{2}$.
That is, $t = \Theta(\frac{(\veps d)^{1/(p-1)}}{p})$.
Since $\delta < \eps / p^{Cp}$, we have $t^{p-1} \geq \delta p^{\Omega(p)} d$. Thus,
$$ \Pr \left[ \sum_{i=1}^d ((X_i+1)^{p-1} - 1) > \veps d \right] \leq 2^{-\Omega(t)}  \leq e^{-\Omega(\frac{(\veps d)^{1/(p-1)}}{p})} \enspace.\qedhere$$
\end{proof}

\section{RIP Construction for $1<p<2$}
In this section, we construct $(k,1+\veps)$-RIP-$p$ matrices for $1 < p < 2$ by proving the following theorem.

We assume that $1 + \tau \leq p \leq 2 - \tau$ for some $\tau > 0$, and whenever we write $O_{\tau}(\cdot)$, we assume that some factor that depends on $\tau$ is hidden. (For instance, factors of $p/(1-p)$ may be hidden.)

\begin{theorem}
\label{rip_upper_12}
For every $n \in \mathbb{Z}_+$, $k\in [n]$, $0 < \eps < 1/2$ and $1 + \tau  \leq p \leq 2 - \tau$,
there exist $m,d\in\mathbb{Z}_+$ with
$$m = O_{\tau}\left(k^p \frac{\log n}{\veps^2} + k^{4-2/p-p} \frac{\log n}{\veps^{2/(p-1)}} \right) \mbox{\quad and \quad} d = O_{\tau}\left(\frac{k^{p-1} \cdot \log n}{\eps} + \frac{k^{(p-1)/p}\cdot \log n}{\veps^{1/(p-1)}} \right)$$
such that, letting $A$ be a random binary $m \times n$ matrix of sparsity $d$, with probability at least $98\%$, $A$ satisfies
$(1-\veps)\|x\|_p^p \leq \|Ax\|_p^p \leq (1+\veps)\|x\|_p^p$
for all $k$-sparse vectors $x\in \mathbb{R}^n$.

Note that, when $k\geq \veps^{-\frac{p(2-p)}{(p-1)^3}}$, the above bounds on $m$ and $k$ can be simplified as
$$m = O_{\tau}\Big(\frac{k^p \cdot \log n}{\veps^2}\Big) \quad\text{and}\quad d=O_{\tau}\Big(\frac{k^{p-1} \cdot \log n}{\veps} \Big) \enspace.$$
\end{theorem}

Our proof of the above theorem is based on the existence of $(\ell, d, \delta)$ bipartite expanders (recall the definition of such expanders from \definitionref{def:bipartite-expander}):
\begin{lemma}{\cite[Lemma 3.10]{bmrv-bo-02}}
    \label{expanders_exist}
    For every $\delta \in (0,\frac{1}{2})$, and $\ell \in [n]$, there exist $(\ell, d, \delta)$-expanders with
    $d = O\big(\frac{\log n}{\delta}\big)$ and $m = O(dl / \delta) =
    O\big(\frac{\ell \log n}{\delta^2}\big)$.
\end{lemma}
In fact, the proof of \lemmaref{expanders_exist} implies a simple
probabilistic construction of such expanders: with probability at least $98\%$, a random binary matrix $A$ of sparsity $d$ is the adjacency matrix of a $(2\ell, d, \delta)$-expander scaled by $d^{-1/p}$, for $\delta = \Theta(\frac{\log n}{d})$ and $\ell=\Theta(\frac{\delta m}{d})$.

Therefore, we will assume that $A$ is the (scaled) adjacency matrix of a $(2\ell, d, \delta)$-expander, for parameters of $\ell$ and $\delta$ that we will specify in the end of this section.%
\footnote{In fact, we will choose $l=\Theta_{\tau}(k^{2-p})$. Therefore, our construction confirms our description in the introduction: it interpolates between the expander construction of RIP-1 matrices from~\cite{bgiks-cgcua-08} that uses $\ell=k$, and the construction of RIP-2 matrices using incoherence argument that essentially corresponds to $\ell=2$.}

\subsection{High-Level Proof Idea}
\label{sec:positive-p-smaller-2:idea}
The goal is to show that $\big|\|Ax\|_p^p - 1\big| \leq \veps$ for every $k$-sparse vector $x$ that satisfies $\|x\|_p=1$. Without loss of generality, let us assume that $x$ is supported on $[k]$, the first $k$ coordinates among $[n]$, and $|x_1| \geq |x_2| \geq \ldots \geq |x_k|$.

\begin{figure}[t]
\center
\includegraphics[width=0.8\textwidth]{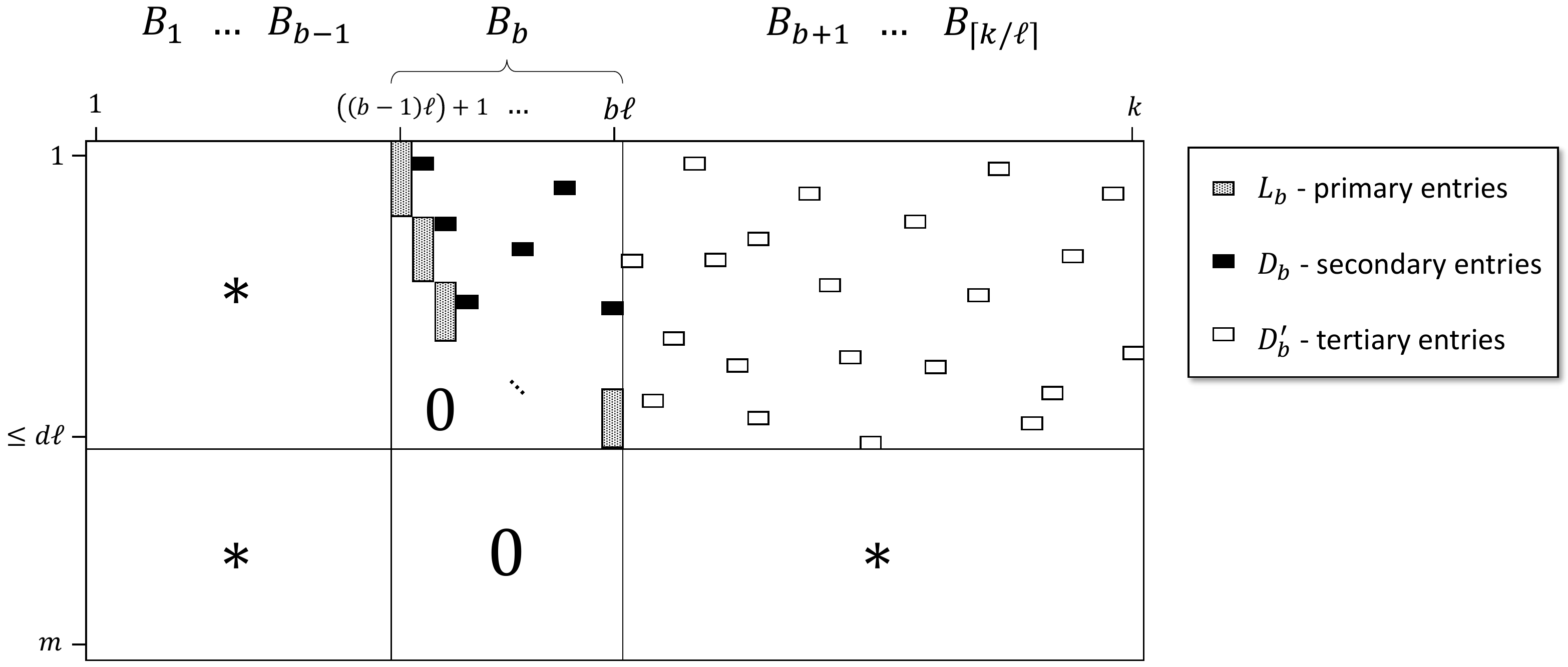}
\caption{Illustrating the definitions of $B_b \subseteq [k]$, $L_b$, $D_b$ and $D_b'$. \label{fig:lead-dots}}
\end{figure}

We partition the $k$ columns into $\lceil k/\ell \rceil$ blocks each of size $\ell$, and denote them by $B_1=\{1,2,\dots,\ell\}$, $B_2 = \set{\ell+1, \ell+2, \ldots, 2 \ell}$, and so on.
    With this definition, we can expand $\|Ax\|_p^p$ as follows:
    \begin{multline}
        \big| \|Ax\|_p^p - 1\big|
        = \Bigg| \sum_{i=1}^m \bigg|\sum_{j=1}^k A_{ij} x_j\bigg|^p - \|x\|_p^p \Bigg|
        = \Bigg| \sum_{i=1}^m \bigg|\sum_{j=1}^k A_{ij} x_j\bigg|^p - \sum_{i=1}^m \sum_{j=1}^k |A_{ij} x_j|^p \Bigg| \\
        \leq O(1) \cdot \sum_{i=1}^m \sum_{j=1}^k \Bigg(|A_{ij} x_j| \cdot \bigg|\sum_{j' = j + 1}^k A_{ij'} x_{j'}\bigg|^{p-1} \Bigg)
        = O(1) \cdot \sum_{b=1}^{\lceil k / \ell \rceil}
\sum_{i=1}^m \sum_{j \in B_b} \Bigg(|A_{ij} x_j| \cdot \bigg|\sum_{j' = j + 1}^k A_{ij'} x_{j'}\bigg|^{p-1}\Bigg) \enspace,
        \label{eqn:main_p_1_2_ineq}
    \end{multline}
    where the inequality follows from \claimref{claim:main_p_1_2_ineq}, a tight bound on the difference between `the $p$-th power of the sum' and `the sum of the $p$-th powers'.

To upper bound the right-hand side of \equationref{eqn:main_p_1_2_ineq},
we fix a block $B_b=\{(b-1)\ell+1,\dots,b\ell\}$
and consider three groups of non-zero entries of $A$: \emph{`primary'},
\emph{`secondary'} and \emph{`tertiary'} entries.

Let us first define primary and secondary entries: together they form a partition of non-zero entries in
the columns of the block $B_b$.
We define \emph{primary} entries $L_b \subseteq [m] \times B_b$ using the following procedure.
For every row of $A$ that has non-zero entries in the columns of $B_b$, we pick the non-zero
entry with the smallest column index and add it to the set of primary entries $L_b$.
We define \emph{secondary} entries $D_b \subseteq [m] \times B_b$ to be the remaining non-zero entries in the columns of  $B_b$.
Finally, we define \emph{tertiary} entries $D_b' \subseteq [m]
\times (B_{b+1} \cup \ldots \cup B_{\lceil k/l\rceil})$ as the set of non-zero entries that lie
in the same row as some primary entry from $L_b$ and in some block $B_{b'}$ for $b' > b$ (see~\figureref{fig:lead-dots}, where we permute rows of $A$ for the sake of illustration).

Next, let us sketch how we upper bound the right-hand side of~\equationref{eqn:main_p_1_2_ineq}.
First, along the way we use crucially the simple estimate $|x_j| \leq j^{-1/p}$ for every $j \in [k]$.
Second, we upper bound the following partial sum of \equationref{eqn:main_p_1_2_ineq} for each $b$ separately:
$$
\sum_{i=1}^m \sum_{j \in B_b} \Bigg(|A_{ij} x_j| \cdot \bigg|\sum_{j' = j + 1}^k A_{ij'} x_{j'}\bigg|^{p-1}\Bigg) \enspace.
$$
We further decompose this sum with respect to $(i, j)$ that are primary (i.e., in $L_b$) or secondary (i.e., in $D_b$), and notice that the pairs $(i,j')$ are either secondary or tertiary (i.e., in $D_b\cup D_b'$).
The crucial observation in our proof is that the entries in $D_b \cup D_b'$ are very sparse and spread across the columns
due to the expansion property of $A$.
Another observation is that for $L_b$, we have at most $d$ entries per column, so we can control the magnitudes
of $|x_j|\leq j^{-1/p}$ for $(i, j) \in L_b$ fairly well.
Overall, the proof of upper bounding the right hand side of \equationref{eqn:main_p_1_2_ineq} boils down to the careful exploitation of these observations and several applications of
H\"{o}lder's inequality.
The details are somewhat lengthy: in particular, we have to treat the case $b = 1$ separately,
and carefully choose all the parameters.
The rest of this section contains the full analysis of this high-level proof idea.

\subsection{Preliminaries}
\begin{claim}
    \label{claim:main_p_1_2_ineq}
    There exists an absolute positive constant $C > 0$ such that%
    \footnote{In fact, choosing $C=3$ should suffice for this claim, but that will make the proof significantly longer.}
    for every $a, b \in \Rbb$ and $1 \leq p \leq 2$ one has
    \begin{equation}
        \label{weird_lp}
        \big| |a+b|^p - |a|^p - |b|^p\big| \leq C |a| |b|^{p-1} \enspace.
    \end{equation}

\noindent
    Therefore, by induction, we obtain that for every $a_1, a_2, \ldots, a_n \in \Rbb$ and $1 \leq p \leq 2$, it satisfies that
    $$
        \bigg|\sum_{i=1}^n a_i\bigg|^p \in \sum_{i=1}^n |a_i|^p \pm C \cdot \sum_{i=1}^{n-1}
        |a_i| \cdot \bigg|\sum_{j=i+1}^n a_j\bigg|^{p-1} \enspace.
    $$
\end{claim}
\begin{proof}
    If $a = 0$ or $b = 0$, then~\eqref{weird_lp} is true for any $C > 0$.
    Otherwise, by homogeneity we can assume that $a = 1$.
    We prove~\eqref{weird_lp} separately for $b < 0$ and $b > 0$.

    \parhead{The Case $b<0$.}
    First, consider the case $b < 0$.
    Our goal is to prove that
    \begin{equation}
        \label{weird_lp_ratio}
        \frac{1 + |b|^p - |1 + b|^p}{|b|^{p-1}}
    \end{equation}
    is bounded from above by a constant for $b < 0$ and $1 \leq p \leq 2$ (obviously,~\eqref{weird_lp_ratio}
    is non-negative).
    Since~\eqref{weird_lp_ratio} is continuous for these values of $b$ and $p$, it is sufficient
    to prove that~\eqref{weird_lp_ratio} is bounded in each of the following two cases:
    \begin{itemize}
        \item $b < b_1$ and $1 \leq p \leq 2$;
        \item $b_2 < b < 0$ and $1 \leq p \leq 2$,
    \end{itemize}
    where $b_1 < b_2 < 0$ are arbitrary constants (one should think of $b_1$ having a large absolute value and $b_2$ being close to zero).

    First, let us consider the case $b < b_1$.
    Denoting by $x = -b > 0$, we need to prove that for sufficiently large $x$ and $1 \leq p \leq 2$,
    $$
        \frac{1 + x^p - (x-1)^p}{x^{p-1}}
    $$
    is bounded. We have
    $$
        \frac{1 + x^p - (x-1)^p}{x^{p-1}} =
        \frac{1 + x^p - x^p \cdot (1-1/x)^p}{x^{p-1}} \leq
        \frac{1 + x^p - x^p \cdot (1-p/x)}{x^{p-1}} \leq
        p + 1 \enspace,
    $$
    where the first inequality is due to the generalized Bernoulli's inequality,
    and the second inequality holds, if $x$ is sufficiently large.

    Now, let us consider the case $b_2 < b < 0$. Defining $\eps = -b > 0$, we need to prove that for sufficiently small $\eps > 0$ and $1 \leq p \leq 2$,
    $$
        \frac{1 + \eps^p - (1 - \eps)^p}{\eps^{p - 1}}
    $$
    is bounded. We have
    $$
        \frac{1 + \eps^p - (1 - \eps)^p}{\eps^{p - 1}}
        \leq \frac{1 + \eps^p - (1 - p \eps)}{\eps^{p - 1}}
        \leq \frac{(p+1)\eps}{\eps^{p-1}} \leq p + 1,
    $$
    where the first inequality is due to the generalized Bernoulli's inequality,
    the second and third inequalities follow from the fact that $1 \leq p \leq 2$.

\parhead{The Case $b>0$.}
    Let us now handle the case $b > 0$.
    It is sufficient to check that for every $b \geq 0$ and $1 \leq p \leq 2$ we have
    $$
        (1 + b)^p \leq 1 + b^p + p b^{p-1} \enspace.
    $$
    This inequality is trivially true when $b=0$, and therefore, it is enough to check that for every $b > 0$ and $1 \leq p \leq 2$,
    $$
        \frac{\partial}{\partial b} ((1 + b)^p - 1 - b^p - pb^{p-1}) = p (1+b)^{p-1} - pb^{p-1} - p(p-1)b^{p-2} \leq 0 \enspace
    $$
    or equivalently
    $$
        \left(1 + \frac{1}{b}\right)^{p-1} \leq 1 + \frac{p-1}{b} \enspace.
    $$
    But the latter follows from the generalized Bernoulli's inequality.
\end{proof}
\begin{lemma}
    \label{holder_1}
    For every $a, c \in \Rbb_{\geq 0}^N$ and $1 \leq p \leq 2$, we have
    $$
        \sum_{i=1}^N c_i a_i^{p - 1} \leq \|c\|_{1/(2 - p)} \cdot \|a\|_1^{p-1} \enspace.
    $$
    In particular, if $c_1=\dots=c_N=1$, we have
    $$
        \sum_{i=1}^N a_i^{p - 1} \leq N^{2-p}
        \cdot \left(\sum_{i=1}^N a_i\right)^{p - 1} \enspace.
    $$
\end{lemma}
\begin{proof}
    This is just an application of H\"{o}lder's inequality for norms $1 / (2 - p)$ and $1 / (p-1)$.
\end{proof}

\subsection{From Expansion Property to the Primary-Secondary-Tertiary Decomposition}
Using the notation from~\sectionref{sec:positive-p-smaller-2:idea}, let us translate the expansion property into a cardinality upper bound on the sets of secondary and tertiary entries.

    \begin{lemma}
        \label{expansion}
        For every integer $1 \leq b \leq \lceil k / \ell\rceil$, we have for every integer $t \geq 1$,
        \begin{equation}
            \label{prefix_expansion}
            \Big|\set{(i, j) \in D_b \cup D'_b \mid j \leq (b-1)\ell+t}\Big| \leq
            \begin{cases}
                0, & t = 1,\\
                3 \delta dt, & t > 1.
            \end{cases}
        \end{equation}
        In addition, we have $|D_b| \leq \delta d \ell$.
    \end{lemma}
    \begin{proof}
        First of all, $|D_b| \leq \delta d \ell$ is an immediate corollary of the expansion property.
        Recall that $A$ is the (scaled) adjacency matrix of a $(2\ell, d,\delta)$-expander and therefore $|D_b| = d|B_b| - |\bigcup_{j\in B_b} S_j| \leq d|B_b| - (1-\delta)d|B_b| = \delta d |B_b| \leq \delta d \ell$.

        \eqref{prefix_expansion} for $t=1$ is obvious, because in the column of $(b-1)\ell+1$, there are only primary entries but not secondary or tertiary ones (see \figureref{fig:lead-dots}).

        For any integer $t$ between $2$ and $\ell$, we observe that the left hand side of \eqref{prefix_expansion} consists only of secondary entries in $D_b$, and moreover,
        $$\Big|\set{(i, j) \in D_b \mid j \leq (b-1)\ell+t}\Big| = dt - \Big| \bigcup_{j=(b-1)\ell + 1}^{j=(b-1)\ell+t} S_j \Big| \leq dt - (1-\delta)dt = \delta dt \leq 3\delta dt \enspace.$$

        For any $t > \ell$, we argue as follows. Since the expander property of $A$ ensures that the union of any $2\ell$ distinct $S_j$'s have at least $(1-\delta)2d\ell$ distinct elements, we conclude that for every $b^* > b$:
        $$
            \Big|\set{(i, j) \in D'_b \mid j \in B_{b^*}}\Big| \leq 2d\ell - \Big| \bigcup_{j\in B_b \cup B_{b^*}} S_j \Big| \leq 2\delta d \ell \enspace.
        $$
        Therefore, for any integer $t>\ell$, suppose that $(b-1)\ell + t\in B_{b'}$ for $b' = b+\lfloor \frac{t-1}\ell \rfloor >b$, then we have
        $$\Big|\set{(i, j) \in D_b \cup D'_b \mid j \leq (b-1)\ell+t}\Big| \leq |D_b| + (b'-b) \cdot 2\delta d\ell \leq \delta d\ell + \frac{t-1}{\ell} \cdot 2\delta d \ell \leq 3 \delta d t \enspace.$$
        This finishes all the cases of \lemmaref{expansion}.
    \end{proof}

The expansion property implies the following useful inequality that will be used extensively in the proof.
    \begin{lemma}
        \label{piotr_modified}
        For every integer $1 \leq b \leq \lceil k / \ell \rceil$, we have
        \begin{equation}
            \label{l1_mass}
            \sum_{(i, j) \in D_b \cup D'_b} A_{ij} |x_j| \leq 3 \delta (dk)^{1-1/p} =: S.
        \end{equation}
        We denote by $S = 3 \delta (dk)^{1-1/p}$ the right-hand side of~\equationref{l1_mass}.
    \end{lemma}
    \begin{proof}%
    \footnote{This is a simple modification of~\cite[Lemma 9]{bgiks-cgcua-08}. However, that lemma does not directly apply for our scenario, because it assumes $A$ expanding any subsets of size at most $k$.}
        Since each non-zero entry of $A$ equals to $d^{-1/p}$, the left hand side of the desired inequality is
        $$d^{-1/p} \cdot \sum_{(i, j) \in D_b \cup D'_b} |x_j| = d^{-1/p} \cdot \sum_{j \geq (b-1)\ell+1} |x_j| \cdot \Big|\set{i \mid (i, j) \in D_b \cup D'_b} \Big| \enspace.$$

        Let us denote by $a_j = \big|\set{i \mid (i, j) \in D_b \cup D'_b} \big|$, the number of distinct nonzero elements in the $j$-th column of $A$ that share rows with the primary entries $L_b$ of the block $b$.
        Then, the above sum equals to
        $$d^{-1/p} \cdot \sum_{j \geq (b-1)\ell+1} |x_j| \cdot a_j = d^{-1/p} \cdot \sum_{t\geq 1} \big|x_{(b-1)\ell+t}\big| \cdot a_{ (b-1)\ell+t}\enspace.$$

        We now observe that, $a_{ (b-1)\ell+1} + \cdots + a_{ (b-1)\ell+t} \leq 3\delta d t$ for every $t\geq 1$ according to \lemmaref{expansion}, while at the same time, $\big|x_{(b-1)\ell+t}\big|$ is assumed to be non-increasing as $t$ increases. Therefore, it one can see that the right hand side of the above sum is maximized when
        $$ a_{ (b-1)\ell+1} = \cdots = a_{ (b-1)\ell+t} = \cdots = 3\delta d \enspace,$$
        and therefore, we conclude that
        $$
            \sum_{(i, j) \in D_b \cup D'_b} A_{ij} |x_j| \leq d^{-1/p} \cdot 3 \delta d  \cdot \|x\|_1
            \leq 3 \delta (dk)^{1-1/p} \enspace,
        $$
        where the last inequality follows from the relation between $\ell_1$ and $\ell_p$ norms, that is $\|x\|_1 \leq k^{1-1/p} \cdot \|x\|_p = k^{1-1/p} $.
    \end{proof}

\subsection{Bounding Equation (\ref{eqn:main_p_1_2_ineq}) for $b > 1$}
    The following estimate upper bounds the right hand side of \equationref{eqn:main_p_1_2_ineq} for any block $b\geq 1$, but we will use
    it eventually only for $b > 1$. For $b = 1$, we will need a separate estimate.

    \begin{lemma}
        \label{second_block}
        For every integer $1 \leq b \leq \lceil k / \ell \rceil$, we have
        $$
            \sum_{i=1}^m \sum_{j \in B_b} \left(|A_{ij} x_j| \cdot \bigg|\sum_{j' = j + 1}^k A_{ij'} x_{j'}\bigg|^{p-1}\right)
            \leq ((3\delta d k)^{2-p} + (\delta d)^{2-p} \ell)\cdot \frac{S^{p-1}}{d^{1/p} \cdot ((b-1)\ell+1)^{1/p}},
        $$
        where $S$ is defined in the statement of \lemmaref{piotr_modified}.
    \end{lemma}
    \begin{proof}
        Let us partition the sum of interest into primary and secondary entries:
        \begin{multline}
            \label{decomposition}
            \sum_{i=1}^m \sum_{j \in B_b} \left(|A_{ij} x_j| \cdot \bigg|\sum_{j' = j + 1}^k A_{ij'} x_{j'}\bigg|^{p-1}\right)
            \\
            \leq
            \sum_{(i, j) \in L_b} \left(|A_{ij} x_j| \cdot \bigg|\sum_{j' = j + 1}^k A_{ij'} x_{j'}\bigg|^{p-1}\right)
            +
            \sum_{(i, j) \in D_b} \left(|A_{ij} x_j| \cdot \bigg|\sum_{j' = j + 1}^k A_{ij'} x_{j'}\bigg|^{p-1}\right)
            = : I + I'.
        \end{multline}
        Now, we upper bound $I$ as follows:
        \begin{equation*}
            \sum_{(i, j) \in L_b} \left(|A_{ij} x_j| \cdot \bigg|\sum_{j' = j + 1}^k A_{ij'} x_{j'}\bigg|^{p-1}\right)
            \leq \frac{1}{d^{1/p} ((b-1)\ell+1)^{1/p}} \cdot
            \sum_{(i, j) \in L_b} \bigg|\sum_{j' = j + 1}^k A_{ij'} x_{j'}\bigg|^{p-1},
        \end{equation*}
        where the inequality follows from the fact that $|x_j| \leq \frac{1}{j^{1/p}}$ (since the coordinates of $x$ are sorted in the decreasing order of their absolute values). We observe that we can apply \lemmaref{holder_1} to the sum $\sum_{(i, j) \in L_b} \big|\sum_{j' = j + 1}^k A_{ij'} x_{j'}\big|^{p-1}$, where the outer sum has at most $3 \delta d k$ non-zero terms.%
        \footnote{This holds, since for every $i, j, j'$ such that $(i, j) \in L_b$, $j' > j$ and $A_{ij'} \ne 0$
        we have $(i, j') \in D_b \cup D_b'$ (see \figureref{fig:lead-dots}). Due to \lemmaref{expansion} we have
        $|D_b \cup D_b'| \leq 3 \delta d k$.}
        As a result, we have
\begin{multline}
            \label{leading}
            \sum_{(i, j) \in L_b} \left(|A_{ij} x_j| \cdot \bigg|\sum_{j' = j + 1}^k A_{ij'} x_{j'}\bigg|^{p-1}\right)        \leq \frac{(3\delta d k)^{2-p}}{d^{1/p} ((b-1)\ell+1)^{1/p}} \cdot
            \left(\sum_{(i, j) \in L_b} \sum_{j'=j+1}^k |A_{ij'} x_{j'}|\right)^{p-1}
            \\\leq \frac{(3\delta d k)^{2-p} \cdot S^{p-1}}{d^{1/p} ((b-1)\ell+1)^{1/p}}
\end{multline}
        where the second inequality follows from \lemmaref{piotr_modified}, as $\sum_{(i, j) \in L_b} \sum_{j'=j+1}^k |A_{ij'} x_{j'}| = \sum_{(i, j) \in D_b \cup D'_b} A_{ij} |x_j| $.

        Next, we upper bound $I'$. For $i \in [m]$ we define
        $c_i := |\set{j \mid (i, j) \in D_b}| \leq \ell$.
        We have
        \begin{align*}
            \sum_{(i, j) \in D_b} \left(|A_{ij} x_j| \cdot \bigg|\sum_{j' = j + 1}^k A_{ij'} x_{j'}\bigg|^{p-1}\right)
            &\overset{\text{\ding{172}}}\leq \frac{1}{d^{1/p} ((b-1)\ell+1)^{1/p}} \cdot \sum_{(i, j) \in D_b} \bigg|\sum_{j' = j + 1}^k A_{ij'} x_{j'}\bigg|^{p-1}
            \\
            &\overset{\text{\ding{173}}}= \frac{1}{d^{1/p} ((b-1)\ell+1)^{1/p}} \cdot \sum_{i=1}^m c_i \bigg|\sum_{j' = j + 1}^k A_{ij'} x_{j'}\bigg|^{p-1}
            \\
            &\overset{\text{\ding{174}}}\leq \frac{ \|c\|_{1/(2-p)} }{d^{1/p} ((b-1)\ell+1)^{1/p}} \cdot
            \left(\sum_{(i, j) \in D_b} \sum_{j'=j+1}^k |A_{ij'} x_{j'}|\right)^{p-1} \\
            &\overset{\text{\ding{175}}}\leq \frac{\|c\|_{1/(2-p)} \cdot S^{p-1}}{d^{1/p} ((b-1)\ell+1)^{1/p}}.
        \end{align*}
        Here, inequality \ding{172} follows from the fact that $|x_j| \leq \frac{1}{j^{1/p}}$, equality \ding{173} follows from the definition of $c_i$, inequality \ding{174} follows from \lemmaref{holder_1}, and inequality \ding{175} follows from \lemmaref{piotr_modified}.

        Observe that for every $1 \leq q \leq \infty$ we have
        $\|c\|_q \leq \|c\|_{\infty}^{1 - 1/q} \|c\|_1^{1/q}$.
        From \lemmaref{expansion} we have $\|c\|_1 = |D_b| \leq \delta d \ell$,
        also by definition of $c_i$ we have $\|c\|_{\infty} \leq \ell$. Overall, we obtain
        \begin{equation}
            \label{geometric_mean}
            \|c\|_{1/(2-p)} \leq \|c\|_{\infty}^{p - 1} \cdot \|c\|_1^{2 - p}
            \leq (\delta d)^{2 - p} \cdot \ell \enspace.
        \end{equation}
        We conclude by combining the upper bound on $I'$ and~\eqref{geometric_mean} as follows:
        \begin{equation}
            \label{dots}
            \sum_{(i, j) \in D_b} \left(|A_{ij} x_j| \cdot \bigg|\sum_{j' = j + 1}^k A_{ij'} x_{j'}\bigg|^{p-1}\right)
            \leq \frac{ (\delta d)^{2-p} \ell \cdot S^{p-1}}{d^{1/p} ((b-1)\ell+1)^{1/p}} \enspace.
        \end{equation}
        Combining~\equationref{decomposition},~\equationref{leading} and~\equationref{dots}, we get the desired inequality.
    \end{proof}

\subsection{Bounding Equation (\ref{eqn:main_p_1_2_ineq}) for $b = 1$}
    The following estimate upper bounds the right hand side of \equationref{eqn:main_p_1_2_ineq} for the block $b=1$. It is tighter than that of \lemmaref{second_block}.

    \begin{lemma}
        \label{first_block}
        For $b = 1$ one has
        $$
            \sum_{i=1}^m \sum_{j \in B_b} \left(|A_{ij} x_j| \cdot \bigg|\sum_{j' = j + 1}^k A_{ij'} x_{j'}\bigg|^{p-1}\right)
            \leq O_{\tau}\left(d^{2-p} \cdot \frac{S^{p-1}}{d^{1/p}} + \delta \ell^{1-1/p} k^{(p-1)^2/p} \right) ,
        $$
        where $S$ is from the statement of \lemmaref{piotr_modified}.
    \end{lemma}
    \begin{proof}
        We again decompose the sum according to the primary and secondary entries:
        \begin{multline}
            \label{decomposition_1}
            \sum_{i=1}^m \sum_{j \in B_b} \left(|A_{ij} x_j| \cdot \bigg|\sum_{j' = j + 1}^k A_{ij'} x_{j'}\bigg|^{p-1}\right)
            \\
            \leq
            \sum_{(i, j) \in L_b} \left(|A_{ij} x_j| \cdot \bigg|\sum_{j' = j + 1}^k A_{ij'} x_{j'}\bigg|^{p-1}\right)
            +
            \sum_{(i, j) \in D_b} \left(|A_{ij} x_j| \cdot \bigg|\sum_{j' = j + 1}^k A_{ij'} x_{j'}\bigg|^{p-1}\right)
            = : I + I'.
        \end{multline}

\noindent
        First, let us upper bound $I$.
        \begin{align}
            \sum_{(i, j) \in L_b} \left(|A_{ij} x_j| \cdot \bigg|\sum_{j' = j + 1}^k A_{ij'} x_{j'}\bigg|^{p-1}\right)
            &\overset{\text{\ding{172}}}\leq \bigg( \sum_{(i,j)\in L_b} |A_{ij} x_j|^{1/(2-p)} \bigg)^{2-p} \cdot \bigg(\sum_{(i, j) \in L_b} \sum_{j' = j + 1}^k |A_{ij'} x_{j'}| \bigg)^{p-1}
            \nonumber\\
            &\overset{\text{\ding{173}}}\leq \bigg( \sum_{(i,j)\in L_b} (A_{ij} \cdot j^{-1/p} )^{1/(2-p)} \bigg)^{2-p} \cdot S^{p-1}
            \nonumber\\
            &\overset{\text{\ding{174}}}\leq \frac{1}{d^{1/p}} \cdot \bigg|\sum_{j=1}^{\ell} d \cdot j^{-\frac{1}{p(2-p)}}\bigg|^{2 - p} \cdot S^{p-1}
            \overset{\text{\ding{175}}}\leq \frac{S^{p-1}}{d^{1/p}} \cdot O_{\tau}(d^{2-p}) .
            \label{leading_1}
        \end{align}
        Here, inequality \ding{172} follows from \lemmaref{holder_1}.
        Inequality \ding{173} follows from the fact that $|x_j| \leq \frac{1}{j^{1/p}}$, and \lemmaref{piotr_modified} (since $\sum_{(i, j) \in L_b} \sum_{j'=j+1}^k |A_{ij'} x_{j'}| = \sum_{(i, j) \in D_b \cup D'_b} A_{ij} |x_j| $).
        Inequality \ding{174} follows from the fact that there are at most $d$ primary entries in the $j$-th column of the matrix for each $j \in B_1 = \{1,\dots,\ell\}$.
        Inequality \ding{175} follows from the fact that $\sum_{j=1}^{\ell} j^{-\frac{1}{p(2-p)}} = O_\tau(1)$ when $\frac{1}{p(2-p)}\geq 1 + \Omega_{\tau}(1)$ (which is true because $1 + \tau \leq p \leq 2 - \tau$).

        Next, let us upper bound $I'$. We note that
        \begin{multline*}
\sum_{(i, j) \in D_b} \left(|A_{ij} x_j| \cdot \bigg|\sum_{j' = j + 1}^k A_{ij'} x_{j'}\bigg|^{p-1}\right)
\overset{\text{\ding{172}}}\leq \left( \sum_{(i, j) \in D_b} |A_{ij} x_j| \right) \cdot \max_{(i,j)\in D_b} \bigg|\sum_{j' = j + 1}^k A_{ij'} x_{j'}\bigg|^{p-1} \\
\overset{\text{\ding{173}}}\leq \left( \sum_{(i, j) \in D_b} |A_{ij} x_j| \right) \cdot \left(\frac{\|x\|_1}{d^{1/p}}\right)^{p-1}
\overset{\text{\ding{174}}}\leq \left( \sum_{(i, j) \in D_b} |A_{ij} x_j| \right) \cdot \left(\frac{k^{1-1/p}}{d^{1/p}}\right)^{p-1} .
        \end{multline*}
        Here, inequality \ding{172} is obvious, inequality \ding{173} follows from $\big|\sum_{j'=1}^k A_{ij} x_j \big| \leq \frac{1}{d^{1/p}} \|x\|_1$, and inequality \ding{174} follows from $\|x\|_1 \leq k^{1-1/p}$.
        Since $A$ expands $B_1$, by~\cite[Lemma 9]{bgiks-cgcua-08} we have%
        \footnote{The proof of this is similar to that of \lemmaref{piotr_modified}. In short, $\sum_{(i, j) \in D_b} |A_{ij} x_j| = d^{-1/p} \cdot \sum_{j=1}^{\ell} |x_j| \cdot \big|\set{i \mid (i,j) \in D_b}\big|$. Denoting by $a_j = \big|\set{i \mid (i,j) \in D_b}\big|$, we can rewrite this sum as $d^{-1/p} \cdot \sum_{j=1}^{\ell} a_j |x_j|$. Now, due to the expansion of $A$, we have $a_1+\cdots+a_t \leq \delta dt$ for every $t$; on the other hand, $|x_j|$ is non-increasing as $j$ increases. Overall, we conclude that this sum is maximized when $a_1=\cdots=a_t = \delta d$, and therefore, we obtain $\sum_{(i, j) \in D_b} |A_{ij} x_j| \leq \delta d^{1-1/p} \sum_{j=1}^\ell |x_j| $.}
        $$
            \sum_{(i, j) \in D_b} |A_{ij} x_j| \leq \delta d^{1 - 1/p} \|x_{B_1}\|_1
            \leq  \delta d^{1-1/p} \cdot \sum_{t=1}^\ell t^{-1/p} \leq \delta d^{1-1/p} \cdot \int_{0}^\ell x^{-1/p} \; dx = \frac{\delta(d\ell)^{1-1/p}}{1-1/p} \enspace,
        $$
        where the second inequality follows from $|x_j| \leq j^{-1/p}$.
        In sum, we have
        \begin{equation}
            \label{dots_1}
            \sum_{(i, j) \in D_b} \left(|A_{ij} x_j| \cdot \bigg|\sum_{j' = j + 1}^k A_{ij'} x_{j'}\bigg|^{p-1}\right)
            \leq
            O_{\tau}\left( \delta (d\ell)^{1-1/p} \cdot \left(\frac{k^{1-1/p}}{d^{1/p}}\right)^{p-1} \right)
            =
            O_{\tau}\left( \delta \ell^{1-1/p} k^{(p-1)^2/p} \right).
        \end{equation}
        Finally, combining~\equationref{decomposition_1},~\equationref{leading_1} and~\equationref{dots_1},
        we get the desired inequality.
    \end{proof}

\subsection{Proof of Theorem~\ref{rip_upper_12}}
    Finally, we are ready to prove \theoremref{rip_upper_12}. We begin with a simple claim.

    \begin{claim}
        \label{integral}
        One has
        $$
            \sum_{b=2}^{\lceil k/\ell \rceil} ((b-1)\ell+1)^{-1/p}
            \leq O_{\tau}\left(\frac{k^{1-1/p}}{\ell}\right) \enspace.
        $$
    \end{claim}
    \begin{proof}
        $$
            \sum_{b=2}^{\lceil k/\ell \rceil} ((b-1)\ell+1)^{-1/p}
            \leq
            \int_1^{\lceil k/\ell \rceil} \frac{dx}{((x-1)\ell+1)^{1/p}}
            \leq
            \frac{1}{\ell} \cdot \int_1^{2k} \frac{du}{u^{1/p}}
            \leq O_{\tau}\left(\frac{k^{1-1/p}}{\ell}\right)\enspace. \qedhere
        $$
    \end{proof}

\begin{proof}[Proof of \theoremref{rip_upper_12}]
    Combining~\equationref{eqn:main_p_1_2_ineq}, \lemmaref{second_block}, \lemmaref{first_block}, \claimref{integral}, and that $S = 3 \delta (dk)^{1-1/p}$, we get
    \begin{align*}
        &\;\quad \big|\|Ax\|_p^p - 1 \big|  \\
        &\leq
O(1) \cdot \sum_{b=1}^{\lceil k / \ell \rceil}
\sum_{i=1}^m \sum_{j \in B_b} \Bigg(|A_{ij} x_j| \cdot \bigg|\sum_{j' = j + 1}^k A_{ij'} x_{j'}\bigg|^{p-1}\Bigg)
\tag{using \equationref{eqn:main_p_1_2_ineq}}\\
&\leq \sum_{b=2}^{\lceil k / \ell \rceil} O_{\tau} \bigg(((3\delta d k)^{2-p} + (\delta d)^{2-p} \ell)\cdot \frac{S^{p-1}}{d^{1/p} \cdot ((b-1)\ell+1)^{1/p}} \bigg) +
O_{\tau}\bigg(d^{2-p} \cdot \frac{S^{p-1}}{d^{1/p}} + \delta \ell^{1-1/p} k^{(p-1)^2/p} \bigg)
\tag{using \lemmaref{second_block} and \lemmaref{first_block}} \\
        &\leq O_{\tau}\left(\left(\frac{k^{1-1/p}}{\ell} \cdot \Big((\delta d k)^{2-p} + (\delta d)^{2-p} \ell\Big) + d^{2-p} \right)\cdot \frac{S^{p-1}}{d^{1/p}} + \delta \ell^{1-1/p} k^{(p-1)^2/p}\right)
\tag{using \claimref{integral}}
        \\
        &=
        O_{\tau}\left(\left(\frac{k^{1-1/p}}{\ell} \cdot \big((\delta d k)^{2-p} + (\delta d)^{2-p} \ell\big) + d^{2-p} \right)\cdot \delta^{p-1} \cdot d^{p-2} \cdot k^{(p-1)^2/p} + \delta \ell^{1-1/p} k^{(p-1)^2/p}\right)
        \\
        &=
        O_{\tau}\big(
            \delta k \ell^{-1} +
            \delta k^{p-1} +
            \delta^{p-1} k^{(p-1)^2/p} +
            \delta \ell^{1-1/p} k^{(p-1)^2/p} \big)\enspace.
    \end{align*}
    We want this expression to be at most $\eps$.
    For this, we can set $\ell = \Theta_{\tau}\left(k^{2-p}\right) \geq 1$ (note that we can do so because $p < 2$), and
    $$\delta = \Theta_{\tau}\left( \min \Big\{\frac{\eps}{k^{p-1}}, \frac{\veps^{1/(p-1)}}{k^{(p-1)/p}}\Big\} \right) \enspace.$$
    Above, when deducing that $\delta \ell^{1-1/p} k^{(p-1)^2/p} \leq O(\veps)$, we have used the fact that $\veps<1$.

    Finally, from \lemmaref{expanders_exist} we can choose $d=O(\frac{\log n}{\delta})$ and get the following number of rows:
    $$
        m = O\left(\frac{dl}{\delta}\right) = O\left(\frac{\ell \cdot \log n}{\delta^2}\right)
        = O_{\tau}\left( \max \Big\{ k^p \frac{\log n}{\veps^2}, k^{4-2/p-p} \frac{\log n}{\veps^{2/(p-1)}} \Big\} \right) \enspace.
    $$
    This finishes the proof of \theoremref{rip_upper_12}.
    \end{proof}

\section{Dimension Lower Bounds}
\label{sec:lower-bound-dimension}
In this section, we prove dimension lower bounds for RIP-$p$ matrices.
\begin{theorem}
\label{thm:lower-bound-dimension}
Let $A$ be an $m\times n$ $(k,D)$-RIP-$p$ matrix with distortion $D>1$. Then,
\begin{align*}
&\text{If $1 < p < 2$,} &&\text{either} \quad m \geq \Omega\Big(\frac{(2-p)n}{pD^2}\Big)^{p/2}
	\quad \text{or} \quad m \geq \Omega\Big( \frac{k^p}{D^{2p/(2-p)}} \Big) \enspace, \\
&\text{If $p>2$,} &&\text{either} \quad m \geq \frac{n}{2k} \quad
	\text{or} \quad m \geq \Omega\Big( \frac{k^p}{D^{p^2/(p-2)}}\Big) \enspace.
\end{align*}
\end{theorem}

\subsection{Three Auxiliary Lemmas}
We start with three auxiliary lemmas.
The first one establishes bounds on the sum of $p$-th powers of the entries of $A$.
\begin{lemma}
\label{lem:pthpowercolumn}
For any column $j \in [n]$, the following holds:
$1 \leq \sum_{i=1}^m |A_{i,j}|^p \leq D^p$.
\end{lemma}
\begin{proof}
$\sum_{i=1}^m |A_{i,j}|^p$ can be viewed as $\|Ae_j\|_p^p$.
Now, for each $j \in [n]$, due to the $(k,D)$-RIP-$p$ property, we have
	$1 \leq \|Ae_j\|_p^p \leq D^p$ completing the proof.
\end{proof}
Next, for any $i \in [m]$ and $t\in [n]$, we denote by $b_{i,t}$ be the $t$-th largest absolute value in row $i$, that is, the $t$-th largest value among $|A_{i,1}|,|A_{i,2}|,\ldots,|A_{i,n}|$. The following lemma establishes upper bounds on individual entries of $A$. Its proof relies on the RIP property for a $k$-sparse vector $x\in\{-1,1\}^n$, chosen so that its entries `match' the sign of the entries of $A$.
\begin{lemma}
\label{lem:absolut}
 We have
$$ \max\limits_{i \in [m]} b_{i,t}
				\leq \left\{
                  \begin{array}{ll}
                    D \cdot t^{1/p-1},
                    & \hbox{if $t\leq k$;} \\
                    D \cdot k^{1/p-1},
                    & \hbox{if $t>k$.}
                  \end{array}
                \right.$$
\end{lemma}
\begin{proof}
We first prove the lemma for any $t\leq k$. Consider any fixed row $i' \in [m]$.
Let $x$ be a $t$-sparse vector such that $x_j = \sgn(A_{i',j})$
	if $A_{i',j}$ is one of the $b_{i',1},\ldots,b_{i',t}$
	and $x_j = 0$ otherwise.
Then, the RIP-$p$ property implies that
$$\|Ax\|_p^p = \sum_{i=1}^m \Big|\sum_{j =1}^n A_{i,j} x_j \Big|^{p} \leq D^p \cdot t \enspace.$$
In particular, since it is the sum over $i$ of $m$ non-negative terms, the above inequality also implies that for any specific row $i' \in [m]$:
$$\Big|\sum_{j = 1}^n A_{i',j} x_j \Big|^{p}
= \Big(\sum_{t = 1}^t \big|b_{i',t}\big|\Big)^{p}
\leq D^p \cdot t
\implies
\sum_{t = 1}^t \big|b_{i',t}\big|
\leq D \cdot t^{1/p} \enspace.$$
Since $|b_{i',t}|$ does not increase as $t$ increases, we get $|b_{i',t}| \leq \frac{D \cdot t^{1/p}}{t} = D \cdot t^{1/p-1}$. This finishes the proof for $t = 1,2,\dots,k$.
For $t > k$, we have $|b_{i',t}| \leq |b_{i',k}| \leq D \cdot k^{1/p-1}$.
\end{proof}

Our third lemma below establishes a lower (or upper) bound on the sum of squares of the entries of $A$. The proof of this lemma relies on the RIP property $\|Ax\|_p \approx \|x\|_p$ examined upon a \emph{random} $k$-sparse vector $x$ sampled from the uniform distribution over $\{-1,1\}^k$.
\begin{lemma}
\label{lem:sqsum}
If $1 < p \leq 2$ then $\sum_{i,j} A_{i,j}^2 \geq n \big(\frac{k}{m}\big)^{2/p - 1}$; if $p \geq 2$ then $\sum_{i,j} A_{i,j}^2 \leq n D^2 \big(\frac{k}{m}\big)^{2/p - 1}$.
\end{lemma}
\begin{proof}
Let $U$ be any set of $k$ distinct indices in $[n]$ (i.e., columns).
Let $\mathcal{X}$ be the distribution of vectors $x \in \mathbb{R}^n$ such that
	$x_i = 0$ if $i \not \in U$ and
	otherwise $x_i$ is an independent random variable attaining values $1$ and $-1$ with probability $1/2$ each.
The $(k,D)$-RIP-$p$ property implies that $k \leq \| Ax \|_p^p \leq D^p k$.

Let us now evaluate the following expectation:
\begin{equation}
\label{eq:expnormp}
\EX\limits_x[ \| Ax \|_p^p ]
	= \sum\limits_{i=1}^m \EX\limits_x\big[ |\langle A_i,x \rangle| ^p\big]
	= \sum\limits_{i=1}^m \EX\limits_x\big[\big(\langle A_i,x \rangle^2\big)^{p/2}\big]
	= \sum\limits_{i=1}^m \sum\limits_{x \in \mathcal{X}}\frac{(\langle A_i,x \rangle^2)^{p/2}}{2^k}
\end{equation}
Comparing the $(p/2)$-th power mean to the arithmetic mean, we have that if $p \leq 2$ then
	$\sum_{x \in \mathcal{X}}\frac{(\langle A_i,x \rangle^2)^{p/2}}{2^k}
		\leq \big(\frac{\sum_{x \in \mathcal{X}} \langle A_i,x \rangle^2}{2^k}\big)^{p/2}$,
and if $p \geq 2$ then
	$\sum_{x \in \mathcal{X}}\frac{(\langle A_i,x \rangle^2)^{p/2}}{2^k}
		\geq \big(\frac{\sum_{x \in \mathcal{X}} \langle A_i,x \rangle^2}{2^k}\big)^{p/2}$.
Finally, because of the way we have defined the distribution $\mathcal{X}$,
	we have $\sum_{x \in \mathcal{X}} \langle A_i,x \rangle^2 = 2^k \cdot \sum_{j \in U} A_{i,j}^2$.

Combining~\eqref{eq:expnormp} with the above pieces for $1<p \leq 2$ gives:
$$k \leq \EX\limits_x[ \| Ax \|_p^p ] \leq \sum\limits_{i=1}^m \big(\sum\limits_{j \in U} A_{i,j}^2 \big)^{p/2} \enspace.$$
\noindent On the other hand, for $p\geq 2$, we get
$$\sum\limits_{i=1}^m \big(\sum\limits_{j \in U} A_{i,j}^2 \big)^{p/2} \leq \EX\limits_x[ \| Ax \|_p^p ] \leq D^p k \enspace.$$

Let us first focus on the case of $1<p\leq 2$.
By again comparing the $(p/2)$-th power mean to the arithmetic mean we can extend our inequality to
\begin{equation}
\label{eq:sqpless1}
\frac{k}{m} \leq \frac{\sum_{i=1}^m \big(\sum_{j \in U} A_{i,j}^2\big)^{p/2}}{m} \leq
\Big(\frac{\sum_{i=1}^m  \sum_{j \in U} A_{i,j}^2}{m}\Big)^{p/2}
\implies \sum_{i=1}^m  \sum_{j \in U} A_{i,j}^2 \geq \frac{k^{2/p}}{m^{2/p-1}} \enspace.
\end{equation}
Enumerating over all possible choices of indices $U$, we get the desired result:
$$ {n \choose k} \cdot \big( \sum_{i=1}^m \sum_{j=1}^n A_{i,j}^2 \big) \cdot \frac{k}{n}
= \sum_U \sum_{i=1}^m \sum_{j \in U} A_{i,j}^2
\geq {n \choose k} \cdot \frac{k^{2/p}}{m^{2/p-1}} \enspace.$$
If $p\geq 2$ then analogously to~\eqref{eq:sqpless1} we get the following inequality:
$$\Big(\frac{\sum_{i=1}^m  \sum_{j \in U} A_{i,j}^2}{m}\Big)^{p/2} \leq
\frac{\sum_{i=1}^m \big(\sum_{j \in U} A_{i,j}^2\big)^{p/2}}{m} \leq
\frac{D^p k}{m}
\implies \sum_{i=1}^m  \sum_{j \in U} A_{i,j}^2 \leq \frac{D^2 k^{2/p}}{m^{2/p-1}} \enspace,$$
and after enumerating over all possible sets of indices $U$ gives:
$\sum_{i=1}^m \sum_{j=1}^n A_{i,j}^2 \leq nD^2 \big(\frac{k}{m}\big)^{2/p-1}$.
\end{proof}

\subsection{Proof of Theorem~\ref{thm:lower-bound-dimension}}
We first focus on the case of $1<p<2$.
Using~\lemmaref{lem:absolut} we can evaluate
\begin{equation}
\label{eq:firstkb}
\sum_{i=1}^m \sum_{j=1}^k b_{i,j}^2
\leq \sum_{i=1}^m \sum_{j=1}^k \left(D \cdot j^{1/p-1}\right)^2
= m D^2 \sum_{j=1}^k j^{2/p-2}
\leq m D^2 \int_{j=0}^{k} j^{2/p-2} dj
\leq O\Big(\frac{p}{2-p} m D^2 k^{2/p-1}\Big)
\end{equation}
\noindent and for the remaining terms:
\begin{align}
\sum_{i=1}^m \sum_{j=k+1}^n b_{i,j}^2
&\overset{\text{\ding{172}}}\leq \sum_{i=1}^m \sum_{j=k+1}^n |b_{i,j}|^p
	\left(D \cdot k^{1/p-1}\right)^{2-p}
\leq \left(\sum_{i=1}^m \sum_{j=1}^n |b_{i,j}|^p\right)
	\left(D \cdot k^{1/p-1}\right)^{2-p} \nonumber \\
&\overset{\text{\ding{173}}}= \left(\sum_{i=1}^m \sum_{j=1}^n |A_{i,j}|^p\right)
	\left(D \cdot k^{1/p-1}\right)^{2-p}
\overset{\text{\ding{174}}}\leq n \cdot D^p \left(D \cdot k^{1/p-1}\right)^{2-p} \enspace, \label{eq:lineusel1}
\end{align}
where \ding{172} follows from \lemmaref{lem:absolut}, \ding{173} follows from the definition of $b_{i,j}$, and \ding{174} follows from~\lemmaref{lem:pthpowercolumn}.
Adding~\eqref{eq:firstkb} and~\eqref{eq:lineusel1} gives:
$$\sum_{i,j} A_{i,j}^2
= \sum_{i=1}^m \sum_{j=1}^n b_{i,j}^2
\leq O\Big(\frac{p}{2-p} m D^2 k^{2/p-1}\Big) +  n \cdot D^p \left(D \cdot k^{1/p-1}\right)^{2-p}$$
\noindent and using~\lemmaref{lem:sqsum} we conclude that:
\begin{align*}
n \cdot \Big(\frac{k}{m}\Big)^{2/p - 1} \leq
O\Big(\frac{p}{2-p} m D^2 k^{2/p-1}\Big) +  n \cdot D^p \left(D \cdot k^{1/p-1}\right)^{2-p} \enspace.
\end{align*}
Therefore, either $n \cdot D^p \big(D \cdot k^{1/p-1}\big)^{2-p}$
or $\frac{p}{2-p}m D^2 k^{2/p-1}$ must be at least $\Omega \big( n \cdot \big(\frac{k}{m}\big)^{2/p - 1} \big)$.
These two cases exactly correspond (after rearranging terms) to the desired inequalities.

(We remark here that when $p=2$, the factor $\big(\frac{k}{m}\big)^{2/p-1}$ on the left hand side becomes $1$, and therefore no interesting lower bound on $m$ can be deduced.)

Next, we focus on the case of $p>2$.
Let us compute again using~\lemmaref{lem:absolut}:
\begin{align*}
\sum_{i=1}^m \sum_{j=k+1}^n b_{i,j}^2
&\geq \sum_{i=1}^m \sum_{j=k+1}^n |b_{i,j}|^p
	\left(D \cdot k^{1/p-1}\right)^{2-p}
\geq \left(\sum_{i=1}^m \sum_{j=k+1}^n |b_{i,j}|^p\right)
	\left(D \cdot k^{1/p-1}\right)^{2-p} \enspace.
\end{align*}

Now, recall that the entries $\{b_{i,j}\}_{i,j}$ are by definition renamed from the entries of $A$, so the summation $\sum_{i=1}^m \sum_{j=k+1}^n |b_{i,j}|^p$ is missing precisely $km$ entries from $A$.
Therefore, this sum contains the $p$-th powers of all of the entries from at least $n-mk$ full columns of $A$, which is at least $n-mk$ (since any full column $j$ of $A$, by~\lemmaref{lem:pthpowercolumn}, has its $p$-th power summing up to at least $1$).
Plugging this into the above inequality we get:
$$\sum_{i=1}^m \sum_{j=k+1}^n b_{i,j}^2
	\geq (n-km) \left(D \cdot k^{1/p-1}\right)^{2-p} \enspace.$$
\noindent On the other hand,
$$\sum_{i,j} A_{i,j}^2
= \sum_{i=1}^m \sum_{j=1}^n b_{i,j}^2
\geq \sum_{i=1}^m \sum_{j=k+1}^n b_{i,j}^2
\geq (n-km) \left(D \cdot k^{1/p-1}\right)^{2-p} \enspace,$$
\noindent and using~\lemmaref{lem:sqsum} we conclude that
$$n D^2 \cdot \Big(\frac{k}{m}\Big)^{2/p - 1}
\geq (n-km) \left(D \cdot k^{1/p-1}\right)^{2-p} \enspace.$$
Now, we either have $m \geq \frac{n}{2k}$ or
\begin{align*}
D^2 \cdot \Big(\frac{k}{m}\Big)^{2/p - 1} \geq \Omega \Big( \left(D \cdot k^{1/p-1}\right)^{2-p} \Big)
&\implies \big(\frac{m}{k}\big)^{(p-2)/p}
	\geq \Omega \Big( \frac{k^{(p-1)(p-2)/p}}{D^p} \Big) \\
&\implies
m \geq \Omega \Big( \frac{k^p}{D^{p^2/(p-2)}}\Big) \enspace.
\end{align*}
Again, we emphasize that we used the strict inequality $p>2$ in the above implication.
\qed

\section{Column Sparsity Lower Bound}
Below we provide a simple lower bound of $\Omega(k^{p-1})$ on the column sparsity of RIP-$p$ matrices.
The proof is a simple extentsion of an argument from~\cite{c-sgccs-10}.
We remark that we are aware of an alternative proof of a slightly stronger lower bound that extends the argument
of Nelson and Nguy$\tilde{\hat{\mbox{e}}}$n~\cite{nn-slbdr-13},
but since the better bound does not seem to be optimal, and the argument is much more complicated,
we decided not to include its proof here.
\begin{theorem}
    Let $A$ be an $m \times n$ matrix with $(k, D)$-RIP-$p$ property and column sparsity $d$.
    Then, either $m > n / k$, or $d \geq k^{p-1} / D^p$.
\end{theorem}
\begin{proof}
    Assume that $m \leq n/k$.
    Since for every basis vector $e_j \in \Rbb^n$ we have $\|Ae_j\|_p \geq 1$,
    it implies that for every column of $A$ there is an entry with absolute value at least $d^{-1/p}$.
    Thus, there exists a row with at least $n / m \geq k$ such entries.
    Without loss of generality, let us assume that this is the first row, and the entries are located in columns from $1$ to $k$.
    There exists a $k$-sparse vector $x$ such that
    \begin{itemize}[nolistsep, topsep=3pt]
        \item for every $1 \leq j \leq k$ we have $x_j = \sgn(A_{1j}) \in \set{-1, 1}$;
        \item for every $j > k$ we have $x_j = 0$;
        \item the first coordinate of the vector $Ax$ is at least $\frac{k}{d^{1/p}}$.
    \end{itemize}
    By the RIP property, we have $\frac{k}{d^{1/p}} \leq \|Ax\|_p \leq D \cdot \|x\|_p = D \cdot k^{1/p}$. Thus,
    $d \geq k^{p-1} / D^p$.
\end{proof}

\section*{Acknowledgments}

We thank Piotr Indyk for encouraging us to work on this project and for many valuable conversations.
We are grateful to Piotr Indyk and Ronitt Rubinfeld for teaching ``Sublinear Algorithms'', where
parts of this work appeared as a final project.
We thank Art\={u}rs Ba\v{c}kurs, Chinmay Hegde, Gautam Kamath,
Sepideh Mahabadi, Jelani Nelson, Huy Nguy$\tilde{\hat{\mbox{e}}}$n, Eric Price
and Ludwig Schmidt
for useful conversations and feedback. Thanks to Leonid Boytsov for pointing us to~\cite{n-iltt1-69,n-iltt2-69}.
We are grateful to anonymous referees for pointing out some relevant literature.
The first author is partly supported by a Simons Graduate Student Award under grant no. 284059.

{
    \bibliographystyle{alpha}
    \bibliography{desk}
}

\appendix
\iflong\else\counterwithin{theorem}{section}\fi

\section{RIP Matrices and Stable Sparse Recovery}
\label{appendix_a}

In this section we extend the main result from \cite{crt-ssrii-06}
to the case of the general $\ell_p$ norms.
Namely, we show that RIP-$p$ matrices for $p > 1$ give rise to the polynomial-time stable sparse recovery
with $\ell_p / \ell_1$ guarantee and approximation factors $C_1 = O(k^{-1 + 1/p})$ and $C_2 = O(1)$.

Suppose that we are given a sketch $y = Ax + e \in \Rbb^m$ for a signal $x \in \Rbb^n$,
where $A \in \Rbb^{m \times n}$, and $\|e\|_p \leq \varepsilon$.
Our goal is to recover from $y$ a good approximation $\widehat{x}$ to $x$.
One of the standard ways to accomplish this is to solve the following \emph{$\ell_1$-minimization} convex program:
\begin{equation}
    \label{l1_minimization}
    \min_{\widehat{x} \in \Rbb^n} \|\widehat{x}\|_1 \quad \text{such that} \quad \|A \widehat{x} - y\|_p \leq \eps \enspace.
\end{equation}

Let $S \subseteq [n]$ be the set of $k$ largest (in absolute value) coordinates of $x$, and $h \defeq \widehat{x} - x$ be the error vector. For a parameter $\alpha > 0$ to be chosen later, we consider the following partition of $[n] \setminus S$:
let $T_0 \subseteq [n] \setminus S$ be the set of $\alpha k$ largest (in absolute value) coordinates of $h$,
let $T_1 \subseteq [n] \setminus (S \cup T_0)$ be the set of $\alpha k$ next largest coordinates, and so on. We state and prove some simple claims first that are true for \emph{every} measurement matrix $A$.

\begin{claim}
    \label{tube_constraint}
    $$
        \|Ah\|_p \leq 2 \eps
    $$
\end{claim}
\begin{proof}
    $$
        \|Ah\|_p = \|A\widehat{x} - Ax\|_p \leq \|A \widehat{x} - y\|_p + \|Ax - y\|_p \leq 2 \eps \enspace,
    $$
    since $\widehat{x}$ is a feasible solution of~(\ref{l1_minimization}), and $\|Ax - y\|_p = \|e\|_p \leq \eps$.
\end{proof}

\begin{claim}
    \label{tail_optimality}
    We have
    $$
        \|h_{\overline{S}}\|_1 \leq \|h_S\|_1 + 2 \|x_{\overline{S}}\|_1 \enspace.
    $$
\end{claim}
\begin{proof}
    Since $x$ is a feasible solution for~\eqref{l1_minimization},
    we have
    $$
        \|x_S\|_1 + \|x_{\overline{S}}\|_1 = \|x\|_1 \geq \|\widehat{x}\|_1 = \|x + h\|_1 \geq
        \|x_S\|_1 - \|h_S\|_1 + \|h_{\overline{S}}\|_1 - \|x_{\overline{S}}\|_1 \enspace. \qedhere
    $$
\end{proof}

\begin{claim}
    \label{sum_of_tails}
    For every $1 \leq p \leq \infty$ we have
    $$
        \sum_{i \geq 1} \|h_{T_i}\|_p \leq \frac{1}{\alpha^{1-1/p}} \cdot \left(\|h_S\|_p +
        \frac{2 \|x_{\overline{S}}\|_1}{k^{1-1/p}}\right) \enspace.
    $$
\end{claim}
\begin{proof}
    For every $i \geq 2$ we have $\|h_{T_i}\|_{\infty} \leq \|h_{T_{i-1}}\|_1 / (\alpha k)$
    by the definition of $T_i$, which implies $\|h_{T_i}\|_p \leq (\alpha k \cdot \|h_{T_i}\|_{\infty}^p)^{1/p} \leq \|h_{T_{i-1}}\|_1 / (\alpha k)^{1-1/p}$. 		Hence
    \begin{align*}
        \sum_{i \geq 1} \|h_{T_i}\|_p \leq \frac{1}{(\alpha k)^{1-1/p}} \cdot \sum_{i \geq 0} \|h_{T_i}\|_1
        &= \frac{\|h_{\overline{S}}\|_1}{(\alpha k)^{1-1/p}}
        \\ &\leq \frac{\|h_S\|_1 + 2 \|x_{\overline{S}}\|_1}{(\alpha k)^{1-1/p}}
        \leq \frac{1}{\alpha^{1-1/p}} \cdot \left(\|h_S\|_p + \frac{2 \|x_{\overline{S}}\|_1}{k^{1-1/p}}\right) \enspace,
    \end{align*}
    where the second inequality follows from~\claimref{tail_optimality}
    and the third inequality follows from the relation between $\ell_1$ and $\ell_p$ norms, that is, $\|h_S\|_1 \leq k^{1-1/p} \cdot \|h_S\|_p $.
\end{proof}

\begin{claim}
    \label{rem_head}
    For every $1 \leq p \leq \infty$ we have
    $$
        \|h_{\overline{S \cup T_0}}\|_p \leq \frac{1}{\alpha^{1-1/p}} \cdot \left(\|h_S\|_p +
        \frac{2 \|x_{\overline{S}}\|_1}{k^{1-1/p}}\right) \enspace.
    $$
\end{claim}
\begin{proof}
    $$
        \|h_{\overline{S \cup T_0}}\|_p = \Big\|\sum_{i \geq 1} h_{T_i}\Big\|_p
        \leq \sum_{i \geq 1} \|h_{T_i}\|_p \leq
        \frac{1}{\alpha^{1-1/p}} \cdot \left(\|h_S\|_p +
        \frac{2 \|x_{\overline{S}}\|_1}{k^{1-1/p}}\right) \enspace,
    $$
    where the last inequality follows from~\claimref{sum_of_tails}.
\end{proof}

\subsection{RIP-$p$ matrices implies $\ell_p/\ell_1$ recovery}

Here we prove that if $A$ is a matrix with RIP-$p$ property, then the $\ell_1$-minimization in \equationref{l1_minimization} recovers a vector that is close enough to $x$. We begin with an auxiliary estimate.

\begin{lemma}
    \label{rip_tail}
    If $A$ is an $((\alpha + 1) k, D)$-RIP-$p$ matrix for $p>1$ and $1 < D < \alpha^{1-1/p}$, then
    $$
        \|h_{S \cup T_0}\|_p \leq \frac{2D}{\alpha^{1-1/p} - D} \cdot \frac{\|x_{\overline{S}}\|_1}{k^{1-1/p}} \enspace + \frac{2\alpha^{1-1/p}}{\alpha^{1-1/p} - D} \cdot \eps.
    $$
\end{lemma}
\begin{proof}
    \begin{align*}
        2\eps &\overset{\text{\ding{172}}} \geq
        		\|Ah\|_p \geq \|Ah_{S \cup T_0}\|_p - \sum_{i \geq 1} \|Ah_{T_i}\|_p
        \\ &\overset{\text{\ding{173}}} \geq
        		\|h_{S \cup T_0}\|_p - D \cdot \sum_{i \geq 1}\|h_{T_i}\|_p
        \\ &\overset{\text{\ding{174}}} \geq
        \|h_{S \cup T_0}\|_p - \frac{D}{\alpha^{1-1/p}} \cdot \left(\|h_S\|_p +
        \frac{2 \|x_{\overline{S}}\|_1}{k^{1-1/p}}\right)
        \geq
        \left(1 - \frac{D}{\alpha^{1-1/p}}\right) \cdot \|h_{S \cup T_0}\|_p -
        \frac{2 D \cdot \|x_{\overline{S}}\|_1}{(\alpha k)^{1-1/p}},
    \end{align*}
    where the inequality \ding{172} is due to Claim~\ref{tube_constraint} both $x$ and $\widehat{x}$ are feasible for~\eqref{l1_minimization}, inequality \ding{173} holds since $A$ satisfies the RIP-$p$ property
    and inequality \ding{174} is due to~\claimref{sum_of_tails}.
\end{proof}

Now we are ready to extend the result from~\cite{crt-ssrii-06}.
We prove that if a measurement matrix $A$ has RIP-$p$ property for $p > 1$,
then one can perform the stable sparse recovery
with the $\ell_p / \ell_1$ guarantee via $\ell_1$-minimization.

\begin{theorem}
    For every $D > 1$, if $A$ is a $((4D)^{p/(p-1)}k,D)$-RIP-$p$ matrix for some $p>1$, then
    $$
        \|h\|_p \leq \frac{O(1)}{k^{1-1/p}} \cdot \|x_{\overline{S}}\|_1 + O(\eps).
    $$
\end{theorem}
\begin{proof}
    Setting $\alpha = (2D)^{p/(p-1)} > 2$, we have $(4D)^{p/(p-1)}k \geq 2^{p/(p-1)}\cdot \alpha \cdot k > (\alpha+1)k$ and therefore the assumptions in \lemmaref{rip_tail} hold. We proceed as follows.
    \begin{align*}
        \|h\|_p \leq \|h_{S \cup T_0}\|_p + \|h_{\overline{S \cup T_0}}\|_p
        &\overset{\text{\ding{172}}} \leq
        		\|h_{S \cup T_0}\|_p + \frac{1}{\alpha^{1-1/p}} \cdot \left(\|h_S\|_p + \frac{2 \|					x_{\overline{S}}\|_1}{k^{1-1/p}}\right)
        \\ &\leq \left(1 + \frac{1}{\alpha^{1-1/p}}\right) \cdot \|h_{S \cup T_0}\|_p
        + \frac{2 \|x_{\overline{S}}\|_1}{(\alpha k)^{1-1/p}}
        \\ &\overset{\text{\ding{173}}} \leq
        		\left(1 + \frac{1}{\alpha^{1-1/p}}\right) \cdot \left(\frac{2D}{\alpha^{1-1/p} - D} \cdot 					\frac{\|x_{\overline{S}}\|_1}{k^{1-1/p}} + \frac{2 \alpha^{1-1/p}}{\alpha^{1-1/p} - D} \cdot \eps\right)
        		+ \frac{2 \|x_{\overline{S}}\|_1}{(\alpha k)^{1-1/p}}
        	\\ &\overset{\text{\ding{174}}} \leq
        		\frac{O(1)}{k^{1-1/p}} \cdot \|x_{\overline{S}}\|_1 + O(\eps).
    \end{align*}
    Above, inequality \ding{172} follows from~\claimref{rem_head},
    inequality \ding{173} follows from~\lemmaref{rip_tail}
    and the last inequality \ding{174} holds because $\alpha^{1-1/p} = 2D$.
\end{proof}

\subsection{$\ell_p/\ell_1$ recovery implies RIP-$p$ matrices}
Here we present a simple argument that any matrix $A$ with $\|A\|_p \leq 1$ that allows stable sparse recovery
with the $\ell_p / \ell_1$ guarantee (with arbitrarily large $C_1$) \emph{must be} $(k, C_2)$-RIP-$p$.
First, observe that the recovery procedure must map $0 \in \Rbb^m$ to $0 \in \Rbb^n$, as long as $C_1$ is finite. Second, let $x \in \Rbb^n$
be any $k$-sparse signal, and consider a sketch $y = Ax + e$, where $e = -Ax$ (thus, $y = 0$).
Since we must recover $0 \in \Rbb^m$ to $0 \in \Rbb^n$, one has from~\eqref{lp_lq}
$$
    \|x\|_p \leq C_2 \cdot \|e\|_p = C_2 \cdot \|Ax\|_p \enspace.
$$
Combining this inequality with $\|Ax\|_p \leq \|x\|_p$ (which follows from $\|A\|_p \leq 1$), we obtain the result.

%
%

\end{document}